\documentclass[11pt, onecolumn]{article}
\usepackage[top=1in, bottom=1in, left=1.25in, right=1.25in]{geometry}

\usepackage{amsfonts}
\usepackage{graphicx}
\usepackage{color, soul}
\usepackage{stfloats}
\usepackage{amsmath, amssymb, cases}
\usepackage[numbers, sort&compress]{natbib}
\usepackage{bm}
\usepackage{booktabs}
\usepackage{footnote}

\usepackage[amsmath,thmmarks]{ntheorem}
\usepackage{theorem}

\newtheorem{thm}{Theorem}
\newtheorem{cor}{Corollary}
\newtheorem{lem}{Lemma}
\newtheorem{rmk}{Remark}
\newtheorem{prp}{Property}

\theoremheaderfont{\sc}\theorembodyfont{\upshape}
\theoremstyle{nonumberplain}
\theoremseparator{}
\theoremsymbol{\rule{1ex}{1ex}}
\newtheorem{proof}{Proof}

\newcommand{\myfrac}[2]{\frac{\textstyle #1}{\textstyle #2}}
\newcommand{\figurewidth}{0.6\textwidth}

\begin{document}
\title{Block-Sparsity-Induced Adaptive Filter for Multi-Clustering System Identification}

\author{Shuyang~Jiang~and~Yuantao~Gu 
\thanks{The authors are with the Department
 of Electronic Engineering, Tsinghua University, Beijing 100084,
 China (E-mail:\,gyt@tsinghua.edu.cn).}
}

\date{Submitted Oct. 18, 2014}

\maketitle

\begin{abstract}
In order to improve the performance of least mean square (LMS)-based adaptive filtering for identifying block-sparse systems, a new adaptive algorithm called block-sparse LMS (BS-LMS) is proposed in this paper. The basis of the proposed algorithm is to insert a penalty of block-sparsity, which is a mixed $l_{2, 0}$ norm of adaptive tap-weights with equal group partition sizes, into the cost function of traditional LMS algorithm. To describe a block-sparse system response, we first propose a Markov-Gaussian model, which can generate a kind of system responses of arbitrary average sparsity and arbitrary average block length using given parameters. Then we present theoretical expressions of the steady-state misadjustment and transient convergence behavior of BS-LMS with an appropriate group partition size for white Gaussian input data. Based on the above results, we theoretically demonstrate that BS-LMS has much better convergence behavior than $l_0$-LMS with the same small level of misadjustment. Finally, numerical experiments verify that all of the theoretical analysis agrees well with simulation results in a large range of parameters.

\end{abstract}

\section{Introduction}
\label{sec:introduction}
Adaptive filtering has been an important research area that attracts much interest in both theoretical and applied issues for a long time \cite{Haykin}. In many scenarios, the unknown systems to be identified are sparse, which means that most of the entries are zero and only a small number of nonzero coefficients exist in the long impulse response (Fig. \ref{fig:sparseresponse}(a)). The typical sparse systems are digital TV transmission channels \cite{Schreiber} and echo paths \cite{Duttweiler}. Among all kinds of sparse systems, there is a family called clustering-sparse systems or block-sparse systems \cite{ITU}. Distinguished from general sparse systems in which the nonzero coefficients may be arbitrarily located, the impulse response of a block-sparse system consists of one or more clusters, wherein a cluster is a gathering of nonzero coefficients (Fig. \ref{fig:sparseresponse}(b,c)). The acoustic echo path is a typical example of single-clustering sparse systems. In satellite-linked communications, the impulse response of the echo path consists of several long flat delay regions and disperse active regions, which is a representative of multi-clustering sparse systems.

\begin{figure}[t]
\begin{center}
\includegraphics[width=0.8\textwidth]{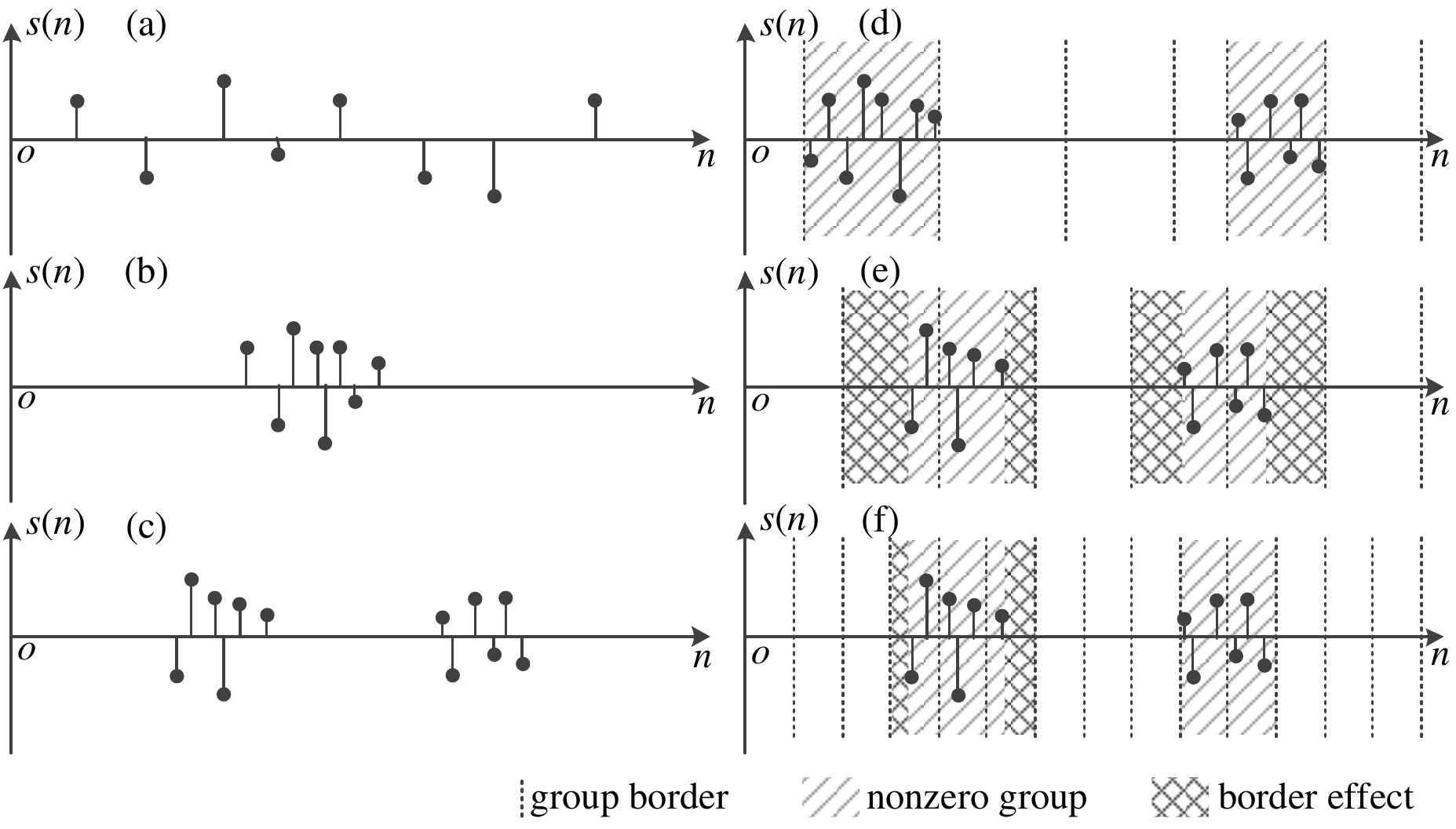}
\caption{(a) A general sparse system. (b) A block-sparse system with one nonzero block. (c) A block-sparse system with two nonzero blocks. (d) The active regions are located randomly in known partition groups. (e) (f) The location of each cluster is arbitrary and unknown, and all of the group partition sizes are the same for practical implementation.}
\label{fig:sparseresponse}
\end{center}
\end{figure}

The least mean square (LMS) algorithm \cite{Widrow} is widely used in various applications due to its low computational cost, easy implementation, and high robustness. However, the traditional LMS has no particular improvement on block-sparse system identification. Many algorithms have been proposed to take advantage of the prior knowledge of block-sparsity. In some algorithms, an auxiliary filter is needed to estimate the positions of the disperse regions. Based on the location information, a number of short adaptive filters are centered at these clusters. The auxiliary filter may be realized as an adaptive delay filter (ADF) \cite{Margo, Margo1} or a full-tap adaptive filter which is operated at a reduced sampling rate \cite{Berg}. In other algorithms, the dispersive regions are detected through the process of convergence. Stochastic Taps NLMS (STNLMS) \cite{Gu} and its two variants \cite{Li, Liu} locate the active region in a stochastic manner. Select and queue with a constraint (SELQUE) algorithm \cite{Sugiyama2} categorizes all taps into two groups: active taps and inactive taps, the latter of which are kept in two queues. The active tap with the minimum absolute coefficient value is replaced by a tap in the queue that is exclusively used for inactive tap indexes residing in the constrained region. An improved M-SELQUE algorithm \cite{Sugiyama3} is applicable to identify an unknown number of multiple dispersive regions. Furthermore, region-based wavelet-packet adaptive algorithm (RBWP) \cite{Nosko} detects the active taps in transform domain and has been shown to be specially effective. In the above mentioned algorithms, the active taps are first located then estimated. The explicitly separated two steps may decelerate the convergence rate and reduce the robustness.

Inspired by a sparsity constraint adaptive algorithm named $l_0$-LMS \cite{Gu2}, we propose block-sparse LMS (BS-LMS) in this work to improve the performance of block-sparse system identification. In $l_0$-LMS, the gradient descent of filter tap-weights are adjusted by approximated $l_0$ norm constraint to learn a general sparse system response. However, it does not utilize the prior knowledge of block-sparsity and has no particular gain when it is applied to identify a clustering sparse system. Motivated by this, we improve $l_0$ norm in the cost function to mixed $l_{2,0}$ norm with equal group partition sizes and exert the sparsity constraint in partitions. We then propose a Markov-Gaussian (M-G) model to generate and describe block-sparse systems. Based on this model, theoretical analysis on the proposed algorithm is conducted. It is proved that BS-LMS outperforms $l_0$-LMS, when the partition size is appropriately chosen. Numerical experiments demonstrate that in block-sparse system identification the proposed algorithm has a faster convergence rate than the reference algorithms with the same steady-state deviation.

This paper is organized as follows. Related works are briefly reviewed in Section \ref{sec:relatedworks}. BS-LMS is proposed in Section \ref{sec:proposedalgorithm}. The theoretical results on steady-state performance and convergence behavior of BS-LMS are presented in Section \ref{sec:theoreticalresult}. The Markov-Gaussian model for generating block-sparse system is proposed and studied in Section \ref{sec:systemmodel}. The key part of this work goes in Section \ref{sec:performanceanalysis}, where the optimal group partition size is studied and  superior performance of BS-LMS compared to $l_0$-LMS is theoretically explained based on the proposed M-G model. Numerical experiments are implemented to verify the above theoretical results in Section \ref{sec:numericalsimulations}. The conclusion is drawn in Section \ref{sec:conclusion}.

\section{Related Work}
\label{sec:relatedworks}

In this section, we briefly review the available (block)-sparsity-constraint-based adaptive algorithms, which are highly relevant to the proposed BS-LMS, from various approaches.

\subsection{Sparsity-Constraint LMS}

The identification of an unknown system with sparse impulse response could be accelerated and enhanced by introducing a sparsity constraint into the cost function of LMS, where the sparsity constraint could be approximated $l_0$ norm \cite{Gu2}, $l_1$ norm \cite{Chen}, reweighted $l_1$ norm \cite{Chen, Taheri}, smoothed $l_0$ norm \cite{Mohimani, Mohimani2}, $l_p$ norm \cite{Wu, Wu2}, or a convex sparsity penalty \cite{Chen2}. However, literature on adaptive filtering algorithms benefiting from block-sparsity is scarce. Thus, it is important to further improve the performance by utilizing block structure. Among the above algorithms, $l_0$-LMS \cite{Gu2} demonstrates rather good performance in experiments and has comprehensive theoretical guarantee \cite{Su2}. Therefore, in this work we generalize $l_0$-LMS to BS-LMS by utilizing block-sparsity. Part of our derivations (mainly in Section \ref{sec:theoreticalresult}) are based on the approach in \cite{Su2}. However, the main contribution of this paper, including BS-LMS algorithm Section (\ref{sec:proposedalgorithm}), the Markov-Gaussian block-sparse model (Section \ref{sec:systemmodel}), and superior performance analysis (Section \ref{sec:performanceanalysis}), are brand-new compared to the above references.

\subsection{Block-Sparse Signal Recovery}

The idea of using mixed norm, such as $l_{2,1}$ norm \cite{Eldar,Stojnic,Stojnic2}, approximated $l_{2, 0}$ norm \cite{Liu2}, $l_{q,1}$ norm \cite{Elhamifar}, to handle block-sparsity has been adopted in sparse signal recovery. By exploiting block structure, recovery may be possible under more general conditions, which demonstrates superior performance brought about by mixed norm. Furthermore, after mixed norm is introduced, the reconstruction error in the presence of noise becomes smaller compared with the conventional algorithms. Besides mixed norm, there are some other approaches in block-sparse signal recovery, including greedy algorithms \cite{Baraniuk,Eldar2,Cevher,Ben-Haim}, Bayesian CS framework-based algorithms \cite{Yu,Zhang}, the dynamic programming-based algorithm \cite{Cevher2} and the decoding-based algorithm \cite{Parvaresh}.

\subsection{Group Sparsity Cognizant RLS}

Recursive least squares (RLS) is another important branch in adaptive filtering. Its faster convergence rate compared to LMS makes RLS an intriguing adaptive paradigm. In \cite{Eksioglu}, group sparsity cognizant RLS is proposed by using various mixed norms, including $l_{2, 1}$ norm, $l_{1, 1}$ norm, $l_{2, 0}$ norm, and $l_{1, 0}$ norm. Numerical experiments show that the novel group sparse RLS is effective and robust for the block-sparse system identification problem, and provides improved performance when compared to the references that only exploit sparsity.

\subsection{Group Partition Selection}

In some of above references \cite{Eldar,Stojnic,Stojnic2,Liu2,Elhamifar,Baraniuk,Eldar2,Ben-Haim,Parvaresh,Eksioglu}, it is assumed that the dispersive active regions are located randomly in known partition groups (Fig. \ref{fig:sparseresponse}(d)). However, one may readily accept that this assumption is impracticable in real scenarios. In fact, the location of each cluster is arbitrary and totally unknown. In this paper, we utilize mixed $l_{2, 0}$ norm in which all of the group partition sizes are the same for practice (Fig. \ref{fig:sparseresponse}(e, f)). Furthermore, in order to avoid the confusion of blocks in unknown system response and the partition blocks in adaptive tap-weights, we adopt \emph{block} or \emph{cluster} to indicate the system coefficient blocks and \emph{group} to denote the partitions in adaptive tap-weights. Based on the theoretical analysis, we will further study the optimal group partition size and demonstrate that the proposed algorithm with an appropriate group partition size achieves superior performance than $l_0$-LMS.

\section{Block-Sparse LMS}
\label{sec:proposedalgorithm}

The proposed algorithm which exploits the block-sparsity of unknown system coefficients is first introduced and then
compared with the available works.

\subsection{Algorithm Description}

The unknown coefficients to be identified and the input signal at time instant $n$ are denoted by ${\bf s} = \left[s_1, s_2, \cdots, s_{L}\right]^{\rm T}$ and ${\bf x}_n = \left[x_n, x_{n-1}, \cdots, x_{n-L+1}\right]^{\rm T}$, respectively, where $L$ is the length of the unknown system and $(\cdot)^{\rm T}$ represents the transposition. The observed output signal is
\begin{equation}
d_n = {\bf x}_n^{\rm T}{\bf s} + v_n,
\label{eq:outputsignal}
\end{equation}
where $v_n$ denotes the measurement noise. The estimated error between the output of the unknown system and that of the adaptive
filter is
\begin{equation}
e_n = d_n - {\bf x}_n^{\rm T}{\bf w}_n,
\label{eq:outputerror}
\end{equation}
where ${\bf w}_n = \left[w_{1,n}, w_{2,n},\cdots, w_{L,n}\right]^{\rm T}$ denotes the adaptive tap-weights.

Motivated by the practical scenarios where the unknown coefficients appear in blocks rather than being arbitrarily spread, we adopt mixed $l_{2, 0}$ norm to evaluate block-sparsity of a vector ${\bf u}=\left[u_1, u_2, \cdots, u_L\right]^{\rm T}$ as
\begin{equation}\label{eq:l20norm}
	\left\| {\bf u} \right\|_{2,0} \triangleq \!\left\| \left[\!\!
	\begin{array}{c}
	\|{\bf u}_{[1]}\|_2\\ \|{\bf u}_{[2]}\|_2\\ \vdots \\ \|{\bf u}_{[N]}\|_2
	\end{array}
	\!\!\right] \right\|_0,
\end{equation}
where ${\bf u}_{[i]} = \left[ u_{(i-1)P+1}, u_{(i-1)P+2}, \cdots, u_{iP}\right]^{\rm T}$ denotes the $i$th group of $\bf u$, $N$ and $P$ denote the number of groups and the group partition size, respectively. We further assume that $L$ can always be divided evenly by $P$ as several zero taps can be added to the tail of ${\bf u}$.

In order to learn the unknown system by utilizing the prior block-sparsity, we design a new cost function, which combines the expectation of the estimated error and mixed $l_{2, 0}$ norm of tap-weight vector,
\begin{equation}
\label{eq:strictcostfunction}
\xi_n \triangleq \textrm{E}\left\{|e_n|^2\right\} + \lambda \left\| {\bf w}_n\right\|_{2,0},
\end{equation}
where $\lambda$ is a positive factor to balance the mean square error and the penalty of block-sparsity. Considering that $l_0$ norm optimization is computationally intractable, we approximate $l_0$ norm in (\ref{eq:strictcostfunction}) by a continuous function \cite{Bradley} and yield
\begin{equation}
\label{eq:costfunction}
\xi_n \approx \textrm{E}\left\{|e_n|^2\right\} + \lambda \sum_{i=1}^{N} \left( 1 - \textrm{exp}\left(-\alpha \left\|{\bf w}_{[i],n}\right\|_2\right)\right),
\end{equation}
where $\alpha$ is a positive constant. One may notice that (\ref{eq:costfunction}) strictly holds when $\alpha$ approaches infinity. By stochastic gradient descent approach and using the approximation
\begin{equation}
{\rm exp}\left(-\alpha |t|\right) \approx
\begin{cases}
1 - \alpha |t|, & |t| \leqslant 1/\alpha; \\
0, & \textrm{elsewhere},
\end{cases}
\end{equation}
the new recursion of the adaptive tap-weights is
\begin{equation}
{\bf w}_{n+1} = {\bf w}_n + \mu e_n {\bf x}_n + \kappa {\bf g}({\bf w}_n),
\label{eq:recurofblockl0LMS}
\end{equation}
where $\mu$ denotes the step-size, $\kappa=\mu \lambda / 2$ adjusts the intensity of block-sparse penalty for given step-size, group zero-point attraction
$$
{\bf g}\!\left({\bf u} \right) \triangleq \big[g_1\!\left({\bf u} \right), g_2\!\left({\bf u} \right), \\ \cdots, g_L\!\left({\bf u} \right)\big]^{\rm T},
$$
and
\begin{equation}
g_k ({\bf u}) \! \triangleq \!
\begin{cases}
2 \alpha^2 u_{k} \! - \! \frac{\textstyle 2 \alpha {u_{k}}}{\textstyle \left\| {\bf u}_{[ \lceil k\!/\!P \rceil]}\right\|_2}, & \! 0 \!<\! \left\| {\bf u}_{[ \lceil k\!/\!P \rceil]}\right\|_2 \!\leqslant\! 1/\alpha; \\
0, & \!\! \textrm{elsewhere},
\end{cases}
\label{eq:gtfunction}
\end{equation}
where $\lceil\cdot\rceil$ denotes ceiling function. In order to avoid being divided by zero, a small positive constant $\delta$ is inserted into the denominator of (\ref{eq:gtfunction}) in real implementation.
The detailed algorithm is described in Table \ref{tab:BSLMS}.

\begin{table}[t]
\renewcommand{\arraystretch}{1.5}
\caption{The Procedure of BS-LMS.} \label{tab:BSLMS}
\begin{center}
\begin{tabular}{l}
\toprule[1pt]
{\bf Input:} \hspace{0.5em} $\{x_n, d_n\}_{n=0, 1, 2, \cdots}$, $L$, $P$, $\mu$, $\alpha$, $\kappa$, $\delta$; \\
{\bf Output:} \hspace{0.5em} $\{{\bf w}_{n}\}_{n=0, 1, 2, \cdots}.$\\
\hline
{\bf Initialization:} \hspace{0.5em} ${\bf w}_0 = {\bf 0}, N=L/P$.\\
{\bf for}  $n = 0, 1, 2, \cdots$\\
\hspace{1.5em} $e_n = d_n - {\bf x}_n^{\rm{T}} {\bf w}_n$; \\
\hspace{1.5em} {\bf for} $i = 1, 2, \cdots, N$ \\
\hspace{3.0em} $E_i = \left(\sum\limits_{j=(i-1)P+1}^{iP} |w_{j,n}|^2\right)^{1/2}$ ; \\
\hspace{1.5em} {\bf end for} \\
\hspace{1.5em} {\bf for} $k = 1, 2, \cdots, L$ \\
\hspace{3.0em} $g_k = 2 \alpha^2 w_{k,n} - 2 \alpha w_{k,n} \textrm{max} \left(\frac{\textstyle1}{\textstyle E_{\lceil k/P \rceil}+\delta}, \alpha \right) $ ; \\
\hspace{3.0em} $w_{k,n+1} = w_{k,n} + \mu e_n x_{n-k+1} + \kappa g_k$;\\
\hspace{1.5em} {\bf end for} \\
{\bf end for} \\
\bottomrule[1pt]
\end{tabular}
\end{center}
\end{table}

\subsection{Relationship with LMS and $l_0$-LMS}
\label{subsec:intuitiveexplanation}

First, one should notice that the group partition size $P$ is a predefined parameter, which is independent of the unknown system to be identified. Here we will discuss two special cases where $P = 1$ and $P = L$.

In the case where $P$ is equal to $1$, mixed $l_{2, 0}$ norm in \eqref{eq:l20norm} is equivalent to
$$
	\left\| {\bf u} \right\|_{2,0} = \!\left\| \left[\!\!
	\begin{array}{c}
	|u_1| \\ |u_2| \\ \vdots \\ |u_L|
	\end{array}
	\!\!\right] \right\|_0 = \|{\bf u}\|_0.
$$
Consequently, the proposed BS-LMS degenerates to $l_0$-LMS because their cost functions are identical. On the other hand, when $P$ is chosen as $L$,  mixed $l_{2, 0}$ norm in \eqref{eq:l20norm} is equivalent to
$$
	\left\| {\bf u} \right\|_{2,0} = \!\left\|
	\left\|{\bf u}\right\|_2 \right\|_0 = \left\{\!\begin{array}{ll}0, & {\bf u}={\bf 0};\\ 1, &{\rm elsewhere.}\end{array}\right.
$$
Therefore, it is readily accepted that BS-LMS degenerates to traditional LMS in this case. Based on the above discussion, one may find that BS-LMS is a generalization of LMS and $l_0$-LMS. Furthermore, the predefined group partition size controls the behavior of BS-LMS. In the next section and afterward, we will discuss how to choose the group partition size for the best performance.

\section{Performance of BS-LMS for General Sparse Systems}
\label{sec:theoreticalresult}

In this section, we follow the study in \cite{Su2} and generalize the theoretical results of $l_0$-LMS to that of the proposed BS-LMS. All conclusions in this section share the similar formulation as their counterparts of $l_0$-LMS, though the constants inside the conclusions are quit different. To save space, the details of assumptions and derivations are omitted, while the \emph{new} constants are listed in Appendix \ref{append:EoC} for reference. However, it should be emphasized that conducting the complicated derivations where the non-unit partition size $P$ is introduced is the main contribution of this section.

\subsection{Assumptions}

Following the approach in \cite{Su2}, we classify the unknown coefficients, correspondingly, the adaptive tap-weights, into three categories in group-partition-wise as
\begin{align*}
\textrm{Large coefficients:}\quad\mathcal{C}_{\rm L}(P) &\triangleq \left\{k\big| \left\|{\bf s}_{[\lceil k\!/\!P\rceil ]}\right\|_2\ge 1/\alpha\right\},\\
\textrm{Small coefficients:}\quad\mathcal{C}_{\rm S}(P) &\triangleq \left\{k\big| 0< \left\|{\bf s}_{[\lceil k\!/\!P\rceil]}\right\|_2 <1/\alpha\right\},\\
\textrm{Zero coefficients:}\quad\mathcal{C}_0(P) &\triangleq \left\{k\big| \left\|{\bf s}_{[\lceil k\!/\!P\rceil ]}\right\|_2 = 0\right\}.
\end{align*}
We further denote the number of tap-weights belonging to the nonzero group partitions by $Q(P)\triangleq|\mathcal{C}_{\rm L}(P)\cup\mathcal{C}_{\rm S}(P)|$, which is also termed the number of nonzero coefficients. However, one should recognize that some zero coefficients may be counted as nonzero coefficients, which is so called \emph{border effect} (Fig. \ref{fig:sparseresponse}(e, f)) and will be studied in next section. Comparing to those defined in \cite{Su2}, one may notice that the above introduced coefficients are closely dependent on the group partition size. Without confusion, however, they are sometimes abbreviated to $\mathcal{C}_{\rm L}, \mathcal{C}_{\rm S}, \mathcal{C}_{0}$, and $Q$.

We could demonstrate that all of the six assumptions in \cite{Su2} still hold because the new recursion does not destroy their validity. We further propose another assumption to make the analysis of BS-LMS feasible.
\begin{enumerate}
\setcounter{enumi}{6}
\item\label{assump:smalldiff}
The difference between the relative strength of $w_{k, n}$ and that of $s_{k}$ in $\mathcal{C}_{\rm S}$ is small enough to ratify the following approximation,
$$\frac{w_{k, n}}{\left\| {\bf w}_{\left[ \left\lceil k/P \right\rceil\right], n} \right\|_2} \approx \frac{s_{k}}{\left\| {\bf s}_{\left[ \left\lceil k/P \right\rceil\right]} \right\|_2}, \quad \forall k \in \mathcal{C}_{\rm S}.$$
\end{enumerate}

This assumption is considered proper due to the following reason. It is readily accepted that in traditional LMS the tap-weights of $w_{k,n}$ uniformly converge to their optimal values with \emph{i.i.d.} white Gaussian input. In the proposed BS-LMS, because of group zero-point attraction in \eqref{eq:recurofblockl0LMS}, the uniform convergence may not exist in a global manner, but may be available inside each group. Therefore, the temporary tap-weight and the unknown coefficient with respect to their strengths in group are supposed very close. In fact, the numerical experiment has verified that this assumption always remains valid, especially in high SNR scenarios.

\subsection{Steady-State Misalignment and Transient Behavior}

Defining ${\bf h}_n \triangleq {\bf w}_n - {\bf s}$ as the misalignment of tap-weights and following the similar approach in \cite{Su2}, the bias in steady state can be derived,
\begin{equation}
\label{eq:meanperformance}
\overline{h_{k, \infty}} \triangleq \lim_{n\rightarrow\infty} \overline{h_{k, n}}= \frac{\kappa}{\mu \sigma_x^2} g_k \left( {\bf s} \right), \quad \forall k=1,2,\cdots,L,
\end{equation}
where \emph{overline} denotes taking expectation and $\sigma_x^2$ denotes the variance of input signal. According to \eqref{eq:gtfunction}, one may find that the tap-weights are unbiased for large and zero group coefficients, while they are biased for small group coefficients.

\begin{lem}[the counterpart of Theorem 1 in \cite{Su2}]
\label{lem:MSD_kappa}
The steady-state mean square deviation (MSD) of BS-LMS is
\begin{align}
D_{\infty} &\triangleq \lim_{n\rightarrow\infty} D_n \triangleq \lim_{n\rightarrow\infty}\overline{{\bf h}_n^{\rm T}{\bf h}_n} \nonumber\\
&= \frac{\mu \sigma_v^2 L}{\Delta_L} + \beta_1 \kappa^2 - \beta_2 \kappa \sqrt{\kappa^2 + \beta_3},\label{eq:Dinftyclose}
\end{align}
where $\sigma_v^2$ denotes the variance of measurement noise, $\Delta_L$ is defined in \eqref{def:DeltaL}, $\left\{ \beta_i \right\}_{i=1,2,3}$ are defined in (\ref{def:beta_1}), (\ref{def:beta_2}), and (\ref{def:beta_3}) in Appendix \ref{append:EoC}, respectively. The step-size should satisfy
\begin{equation}
\label{eq:conditionofmu}
0 < \mu < \mu_{\rm{max}} \triangleq \frac{2}{\left( L + 2\right) \sigma_x^2}
\end{equation}
to guarantee convergence .
\end{lem}

\begin{lem}[the counterpart of Corollary 1 in \cite{Su2}]
\label{lem:MSD_kappa1}
In order to make the steady-state MSD be as small as possible, the best choice for $\kappa$ is
\begin{equation}
\label{eq:kappa_optimal}
\kappa_{\rm{opt}} = \frac{\sqrt{\beta_3}}{2} \left( \sqrt[4]{\frac{\beta_1 + \beta_2}{\beta_1 - \beta_2}} - \sqrt[4]{\frac{\beta_1 - \beta_2}{\beta_1 + \beta_2}}\right)
\end{equation}
and the minimum steady-state MSD is
\begin{equation}
\label{eq:MSD_kappa}
D_{\infty}^{\rm{min}} = \frac{\mu \sigma_v^2 L}{\Delta_L} + \frac{\beta_3}{2} \left( \sqrt{\beta_1^2 - \beta_2^2} - \beta_1 \right).
\end{equation}
\end{lem}

\begin{lem}[the counterpart of Theorem 2 in \cite{Su2}]
\label{lem:MSD_instant}
For a given unknown system, the closed form of instantaneous MSD is
\begin{equation}
\label{eq:MSD_instant_closeform}
D_n = c_1 \lambda_1^n + c_2 \lambda_2^n + c_3 \lambda_3^n + D_{\infty},
\end{equation}
where $\lambda_1$ and $\lambda_2$ are the eigenvalues of matrix ${\bf A}$, which is defined in (\ref{def:A}). $c_1$ and $c_2$ are coefficients defined by initial value (please refer to Lemma 1 in \cite{Su2}). The expressions of constants $\lambda_3$ and $c_3$ are listed in \eqref{def:lambda_3} and \eqref{def:c_3}, respectively, in Appendix \ref{append:EoC}.
\end{lem}

\begin{rmk}
Based on the above lemmas, we have successfully generalize the theoretical results of $l_0$-LMS to BS-LMS. As we have mentioned, most of their formulations are exactly identical, whereas the constants included are rather different. To totally understand the above contents, the readers are recommended to refer to \cite{Su2} and compare those constants in Appendix A of \cite{Su2} with those in Appendix A of this paper. Based on the foundation of this section, we have prepared to comprehensively study the performance of BS-LMS.
\end{rmk}

\section{Markov-Gaussian Model for Generating Block-Sparse Systems}

\label{sec:systemmodel}

Inspired by the characteristic of nonzero (or zero) coefficients clustering in blocks, we propose a Markov-Gaussian (M-G) model with parameter set, $\mathcal{M}(L, p_1, p_2, \sigma_s^2)$, to generate a wide range of block-sparse systems. One will notice that the proposed one is a simplified Ising model that fits to the scenario in this study.

Utilizing the proposed model, the impulse response $\bf s$ of a block-sparse system is generated in two steps. In the first step, the zero and nonzero sets which contain the index of zero coefficients and nonzero coefficients, respectively, are produced by a Markov process. From $1$ to $L$, index $k$ is iteratively and stochastically determined to fall into zero or nonzero sets based on the class of index $(k-1)$. Please refer to Fig. \ref{fig:MarkovChain} for detail, where for $k=2,3,\cdots,L$
\begin{align*}
	{\rm P}\left\{s_k = 0 | s_{k-1} = 0\right\} = p_1,\\
	{\rm P}\left\{s_k \ne 0 | s_{k-1} \ne 0\right\} = p_2,
\end{align*}
and the category of $s_{1}$ is decided by $s_{0}$, which is an imaginative scaler and fixed to zero. In the second step, after the nonzero set is determined, the amplitudes of nonzero coefficients are independently and identically drawn from a Gaussian distribution with zero mean and variance $\sigma_s^2$.

\begin{figure}[t]
\begin{center}
\includegraphics[width=0.6\textwidth]{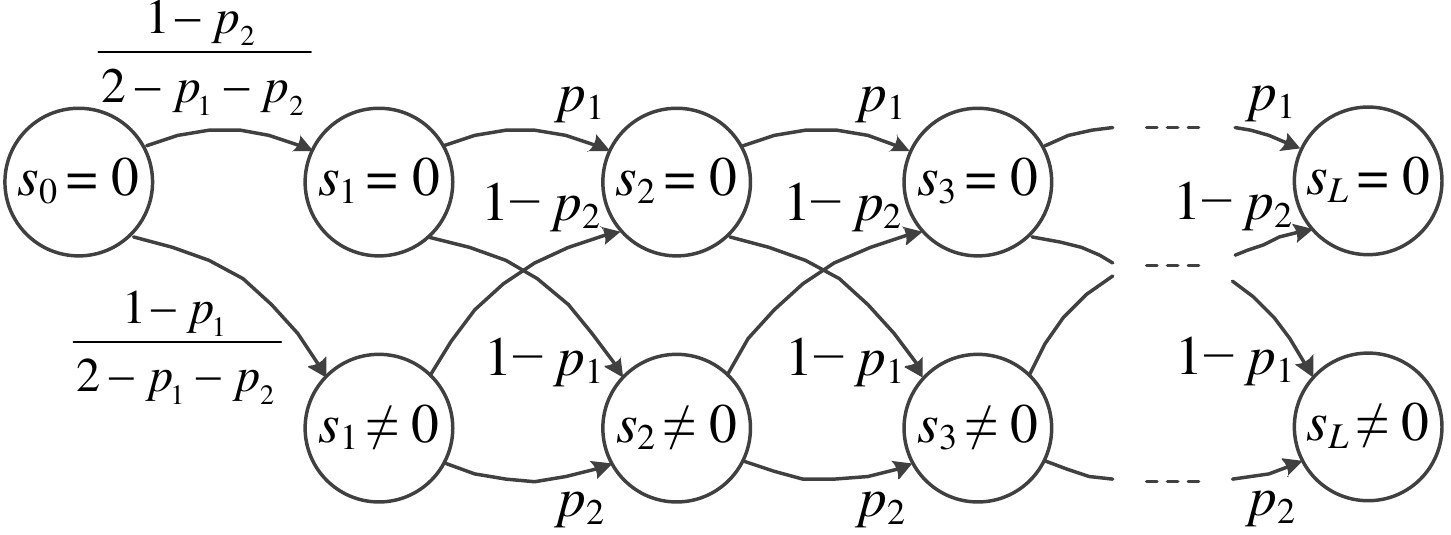}
\caption{The proposed model for generating block-sparse impulse response.}
\label{fig:MarkovChain}
\end{center}
\end{figure}

According to Fig. \ref{fig:MarkovChain}, for the sake of producing a sparse system response, it should be guaranteed that $(1 - p_2)$ is far larger than $(1 - p_1)$. Furthermore, both $p_1$ and $p_2$ need to be very close to $1$ in order to generate a clustering impulse response. The proposed model has several properties that demonstrate its advantages and will be used to analyze the proposed BS-LMS in the following section.

\begin{prp}\label{prp:sparsity}
(\emph{Sparsity and block size}) For given $\mathcal{M}(L, p_1, p_2, \sigma_s^2)$, the average percentage of nonzero coefficients, the average block size of nonzero and zero coefficients of the generated impulse responses, which are denoted by $\overline S, \overline{B_{\rm nz}}$, and $\overline{B_{\rm z}}$, respectively, follow
\begin{align*}
	\overline S = \frac{1 - p_1}{2 - p_1 - p_2}, \quad \overline{B_{\rm nz}} = \frac{1}{1 - p_2},\quad\textrm{and}\quad \overline{B_{\rm z}} = \frac{1}{1 - p_1}.
\end{align*}
\end{prp}

Several examples of block-sparse systems generated by $L=800, \sigma_s^2=1$, and various $(p_1, p_2)$ are showed in Fig. \ref{fig:coeff_markov}. Inside every row and every column, the average block size of zero and nonzero coefficients increase, respectively, with respect to $p_1$ and $p_2$. Moreover, the three responses located on the diagonal subplots satisfy $\overline S = 0.1$ and have $80$ expected nonzero coefficients.

\begin{figure}[t]
\begin{center}
\includegraphics[width=0.8\textwidth]{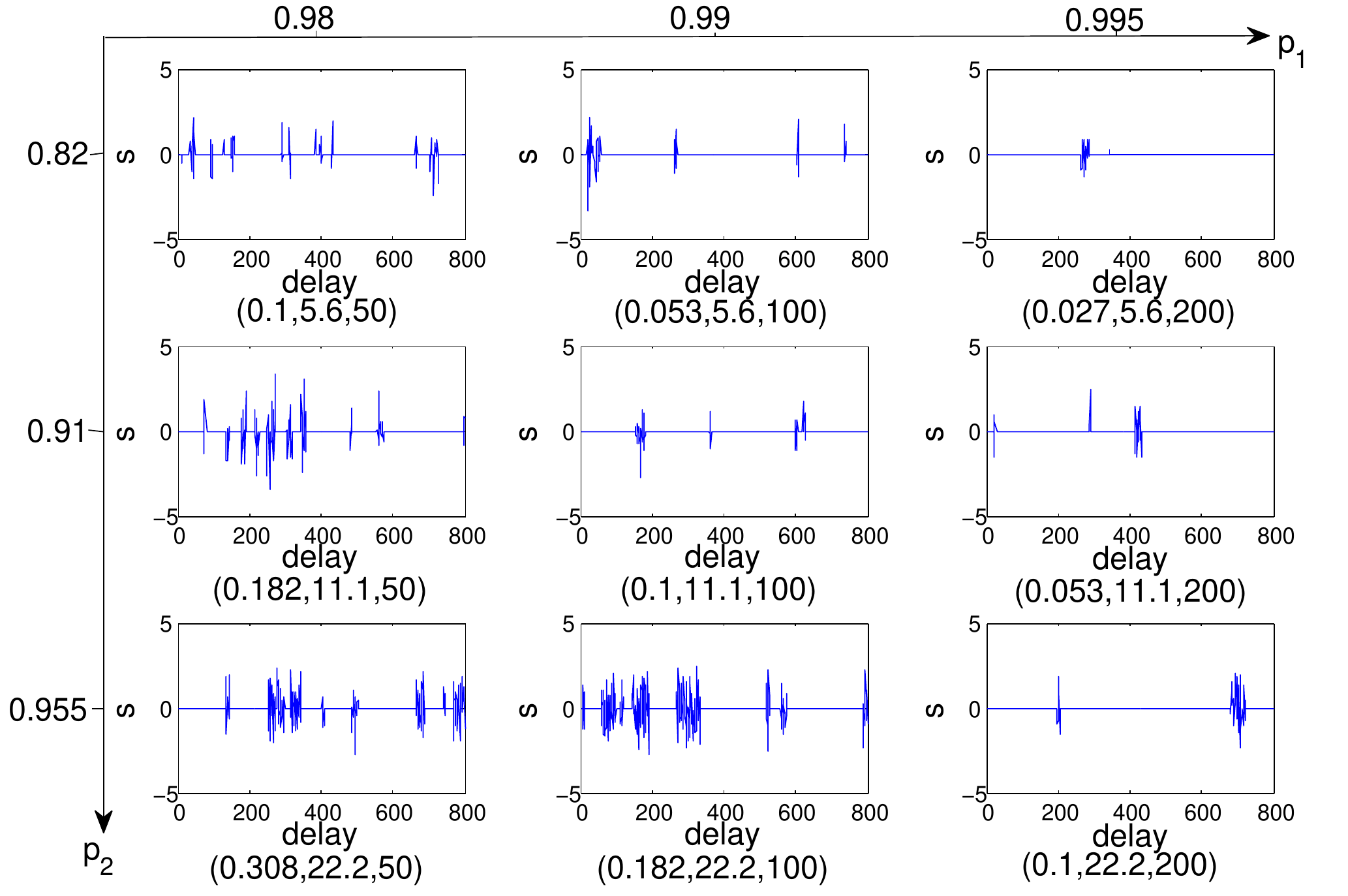}
\caption{The examples of block-sparse impulse responses generated by the proposed M-G model with $L=800, \sigma_s^2=1$ and various $(p_1, p_2)$. $(\overline{S}, \overline{B_{\rm nz}}, \overline{B_{\rm z}})$ with corresponding parameter set is listed below each subfigure. The three systems located on the diagonal subplots share the same sparsity.}
\label{fig:coeff_markov}
\end{center}
\end{figure}

Next we will study the border effect quantitatively based on the proposed model.
\begin{prp}\label{prp:bordereffect}
(\emph{Border effect}) For given $\mathcal{M}(L, p_1, p_2, \sigma_s^2)$ and a predefined group partition size $P$, we can calculate the average number of tap-weights belonging to nonzero groups, $\overline{Q}$, to describe the intensity of the border effect,
\begin{equation}
\label{def:Q_mean}
\overline{Q} = L \left( 1 - (1 - \overline{S}) p_1^{P-1}\right).
\end{equation}
\end{prp}

\begin{proof}
The proof is postponed to Appendix \ref{append:bordereffect}.
\end{proof}

According to this property, one may find that $\overline{Q}$ becomes larger when $P$ increases, which shows the border effect is heavier.

At last, we will show the relationship between the proposed M-G model and the Ising model \cite{McCoy}, which is a prototypical Markov random field.

\begin{prp}\label{prp:isingmodel}
(\emph{Relation with Ising model}) The proposed Markov model for determining zero and nonzero coefficients sets is a special case of the Ising model. Specifically, for the Ising model, its probability density function is
$$
p({\rm sp}({\bf s}); \bm{\zeta}, \bm{\zeta}') = {\rm exp} \left\{ \sum_{i=1}^{L} \zeta_i s_i + \sum_{i=1}^{L - 1} \zeta'_i s_i s_{i+1} - Z_{\bf s}(\bm{\zeta}, \bm{\zeta}')\right\},
$$
where ${\rm sp}({\bf s})$ denotes the support of $\bf s$,
$$
	{\rm sp}(s_i) = \left\{\!\!\!\begin{array}{rl}1, & s_i\ne 0;\\
	-1, & s_i=0,\end{array}\right.
$$
and $Z_{\bf s}(\bm{\zeta}, \bm{\zeta}')$ is a strictly convex function with respect to $\bm{\zeta}$ and $\bm{\zeta}'$ that normalizes the distribution so that it integrates to one. When
\begin{align}
\zeta_i & =
\begin{cases}
\myfrac{1}{4} \ln \myfrac{p_2(1-p_1)}{p_1(1 - p_2)}, & i = 1,L; \\
\myfrac{1}{2} \ln \myfrac{p_2}{p_1}, & i = 2,\cdots,L-1,
\end{cases} \nonumber \\
\zeta'_i & = \frac{1}{4} \ln \frac{p_1 p_2}{(1 - p_1)(1 - p_2)}, \quad i=1,\cdots,L-1, \nonumber
\end{align}
the Ising model degenerates to the proposed Markov model, which is equipped by concise and meaningful parameters.
\end{prp}

As far as we know, this is the first time that Markov process is used to describe block-sparsity. Besides the scenarios of system identification, the proposed M-G model may be utilized in various research area to generate arbitrary system response with given block-sparse constraint.

\section{Performance of BS-LMS for Block-Sparse Systems}
\label{sec:performanceanalysis}

The behavior of BS-LMS in block-sparse scenario is further studied in this section by utilizing the proposed M-G model. New assumptions are adopted as follows.

\begin{enumerate}
\setcounter{enumi}{7}
\item\label{assump:smallP}
For a given unknown system, which is supposed to be long and sparse, the partition size $P$ is small with respect to the filter length to guarantee that the system response in group-partition-wise is still sparse, i.e., $2 \ll Q \ll L$.
\item\label{assump:markovmodel}
The unknown system response to be identified is generated by the proposed M-G model, $\mathcal{M}(L, p_1, p_2, \sigma_s^2)$. Please refer to Fig. \ref{fig:application} for illustration.
\end{enumerate}

\begin{figure}[t]
\begin{center}
\includegraphics[scale=0.6]{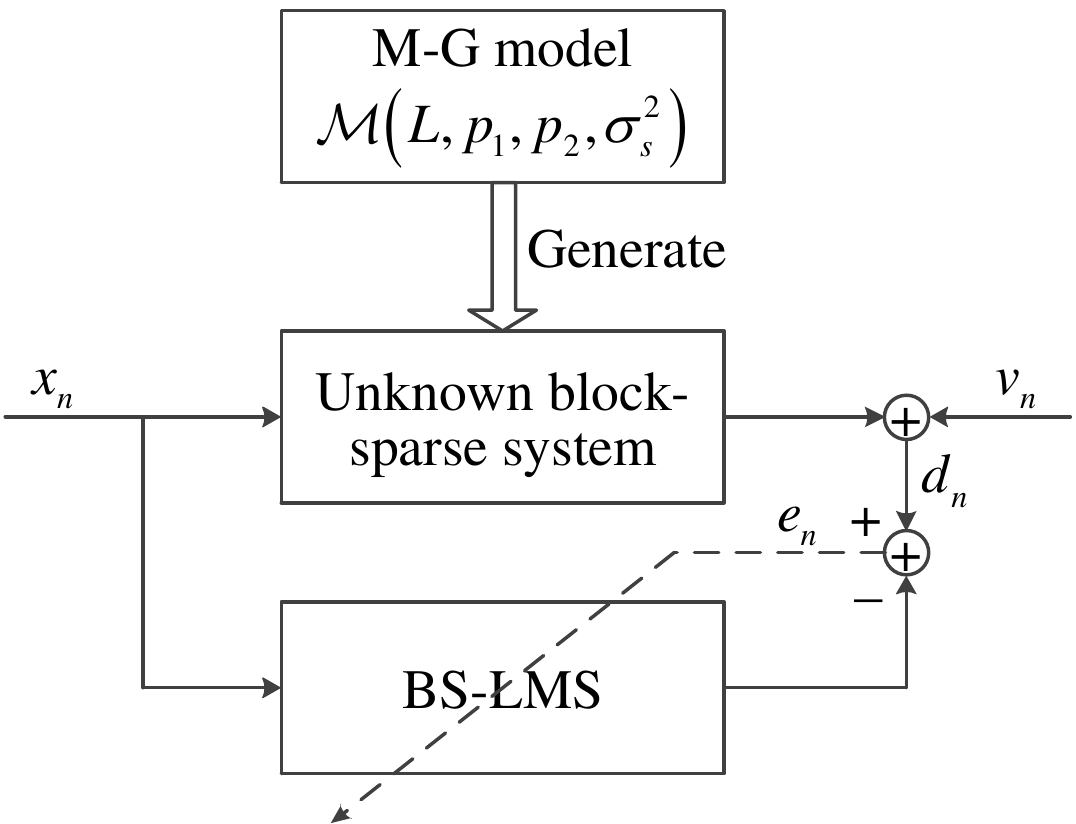}
\caption{The framework of studying the average performance of BS-LMS in identifying block-sparse systems generated by a M-G model.}
\label{fig:application}
\end{center}
\end{figure}

Assumption \ref{assump:smallP}) makes sense because $P$ is an important predefined parameter that need to be elaborately selected. Furthermore, it can be accepted that BS-LMS penalizes sparsity in group-partition-wise and an overlarge $P$ definitely destroys the sparsity.

The introduction of M-G model by Assumption \ref{assump:markovmodel}) makes it feasible to analyze BS-LMS with respect to block-sparse systems. As a consequence, we will study the performance of BS-LMS in the sense of \emph{expected} unknown system response, which is generated by the given M-G model with specified parameters. Therefore, the average minimum steady-state MSD (AMS-MSD), the optimal group partition size, and the average minimum transient MSD (AMT-MSD), denoted as $\overline{D_{\infty}^{\rm min}}, P_{\rm opt}$, and $\overline{D_n^{\rm min}}$, respectively, will be derived in the following text.

Finally we demonstrate that BS-LMS outperforms $l_0$-LMS in convergence rate significantly when the group partition size is chosen close to its optimum.

\subsection{Steady-State Performance and Optimal Group Partition Size}

The following theorem presents the effect of the group partition size on the AMS-MSD and the selection of the optimal partition size.

\begin{thm}
\label{thm:MSD_B}
For given block-sparse systems generated by the proposed M-G model, the AMS-MSD of BS-LMS is
\begin{equation}\label{eq:MSD_B}
\overline{D_{\infty}^{\rm{min}}} \approx \frac{\mu\sigma_v^2}{\Delta_{\overline Q}}\left({\overline Q} + \frac{\sqrt{2 \pi(L - {\overline Q}) \overline{G({\bf s})} }}{\alpha \theta(P) {\Delta_{\overline{Q}}}}\right).
\end{equation}
where $\overline{Q}, \Delta_{\overline{Q}}$, and $\overline{G(\bf s)}$ are defined in \eqref{def:Q_mean}, (\ref{def:mean_Delta_Q}), and (\ref{def:Gamma1_mean}), respectively. The optimal group partition size could be numerically found by
\begin{equation}\label{eq:Popt}
P_{\rm{opt}} = \mathop{\arg}\mathop{\min}_{P} \overline{D_{\infty}^{\rm{min}}}.
\end{equation}
\end{thm}

\begin{proof}
The proof is postponed to Appendix \ref{append:proofofthmMSDB}.
\end{proof}

\begin{cor}\label{cor:MSEproportionalstepsize}
The AMS-MSD monotonically increases with respect to the step-size.
\end{cor}

Corollary \ref{cor:MSEproportionalstepsize} is coincident with the intuition and the theory on $l_0$-LMS. This can be readily seen from \eqref{eq:MSD_B} because $\Delta_{\overline Q}$ monotonically decreases with respect to $\mu$.

\begin{rmk}\label{rmk:comparesteadysate}
As $l_0$-LMS is a special case of BS-LMS when $P$ equals $1$, we expect that the AMS-MSD of BS-LMS with $P_{\rm{opt}}$ is no larger than that of $l_0$-LMS when all the other parameters are same. Considering the high complexity of the closed form of $P_{\textrm{opt}}$, it is not derived here for the sake of simplicity.
\end{rmk}

\subsection{Superior Performance of BS-LMS}

Based on Lemma \ref{lem:MSD_instant} and Theorem \ref{thm:MSD_B}, we could demonstrate that the averaged transient behavior of BS-LMS with an appropriate group partition size is better than that of $l_0$-LMS.

\begin{thm}
\label{thm:meanMSD_instant}
For given block-sparse systems generated by the proposed M-G model, the closed form of the AMT-MSD of BS-LMS is
\begin{equation}
\label{eq:meanMSD_instant_closeform}
\overline{D_n^{\rm min}} = c'_1 (\lambda'_1)^n + c'_2 (\lambda'_2)^n + c'_3 (\lambda'_3)^n + \overline{D_{\infty}^{\rm min}},
\end{equation}
where
\begin{align}
\label{eq:lambda_1}
\lambda'_1  \triangleq & 1 - 2\mu \sigma_x^2, \\
\label{eq:lambda_2}
\lambda'_2 \triangleq & 1 - 2\mu \sigma_x^2 \alpha \theta(P)\sqrt{\frac{2 \left( L - \overline{Q} \right) }{\pi \overline{G({\bf s})}}}, \\
\label{eq:lambda_3}
\lambda'_3  \triangleq & 1 - \mu \sigma_x^2,
\end{align}
and the expressions of $\left\{c'_i\right\}_{i=1,2,3}$ share the same forms with $\left\{c_i\right\}_{i=1,2,3}$ in Lemma \ref{lem:MSD_instant} except for that $Q, G({\bf s}), \| s\|_2^2$, and $G'({\bf s})$are replaced by their means defined by \eqref{def:Q_mean}, \eqref{def:Gamma1_mean}, \eqref{def:norms_mean}, and \eqref{def:Gamma2_mean}, respectively.
\end{thm}

\begin{proof}
The proof is postponed to Appendix \ref{append:proofoflemmaMSDinstant}.
\end{proof}

The difference of Theorem \ref{thm:meanMSD_instant} from Lemma \ref{lem:MSD_instant} is that the former provides an averaged minimum behavior of the best $\kappa$ with respect to a given unknown system generation model. As a consequence, the close form of $\{\lambda'_i\}$ are provided to reveal the detail of convergence in BS-LMS.

In order to compare the convergence behavior of BS-LMS and $l_0$-LMS fairly, it is assumed that the final steady-state MSDs of both algorithms are equal. According to Corollary \ref{cor:MSEproportionalstepsize} and Remark \ref{rmk:comparesteadysate}, we know that the step-size in BS-LMS is larger than that in $l_0$-LMS when the two algorithms demonstrate the same steady-state performance. Then we have the following corollary.

\begin{cor}
\label{cor:MSD_smalllambda}
For a given M-G model $\mathcal{M}(L, p_1, p_2, \sigma_s^2)$ satisfying
\begin{equation}
\label{eq:res_p12}
\frac{1 - p_1}{(1 - p_2)^2} \geq \frac{1}{3(1 - {\rm e}^{-1})},
\end{equation}
$\left\{\lambda'_i\right\}_{i=1,2,3}$ in \eqref{eq:meanMSD_instant_closeform} with $P_{\textrm{opt}}$ are smaller than, respectively, those of $l_0$-LMS, which means that BS-LMS with $P_{\textrm{opt}}$ always converges more quickly than $l_0$-LMS.
\end{cor}

\begin{proof}
The proof is postponed to Appendix \ref{append:proofofcorollaryMSD_smalllambda}.
\end{proof}

Though we further restrict the selection of $p_1$ and $p_2$ to facilitate the proof, it is found that \eqref{eq:res_p12} is not necessary. In fact, Corollary \ref{cor:MSD_smalllambda} is usually valid even when \eqref{eq:res_p12} is violated, as shown in numerical results.

\section{Numerical Simulations}
\label{sec:numericalsimulations}

Five experiments are designed to verify the contents in this paper, where the first two demonstrate the performance of BS-LMS and the M-G model, and the last three test the theoretical analysis. The reference algorithms include STNLMS \cite{Gu}, SELQUE \cite{Sugiyama2}, M-SELQUE \cite{Sugiyama3}, and $l_0$-LMS \cite{Gu2}. In all the experiments, the unknown systems are of length $L=800$. For those generated with the proposed M-G model, $\sigma_s^2$ is set as $1$. The input signal and measurement noise are independent zero mean Gaussian series. For both BS-LMS and $l_0$-LMS, $\alpha$ is chosen as $1$. For BS-LMS, $\delta = 1{\rm e} \textnormal{-}8$. The Signal-to-Noise Ratio (SNR) is 40dB and 20dB in the fifth experiment, while it is 40dB in the others. Simulation results are averaged by 10 independent trials for each unknown system. To get the average MSD, 100 unknown systems are generated and identified, and then the MSDs of these systems are averaged.

\subsection{On the Performance of BS-LMS and M-G Model}

In the first experiment, the proposed algorithm is tested and compared with the references by identifying two block-sparse systems which have the same sparsity. The impulse responses of these systems are displayed in Fig. \ref{fig:traditionalchannels}, where the first has a single cluster of nonzero coefficients at $[405,451]$ and the second has two clusters at $[405, 429]$ and $[569, 590]$, respectively. The simulation results are plotted in Fig. \ref{fig:MSD_traditionalchannels}. For the proposed algorithm, the parameters are set as $\mu=1/L$, $\kappa=1.55{\rm e} \textnormal{-}6$ and $P=5$. For all reference algorithms, the step-sizes are selected to make their steady-state MSDs equal to that of BS-LMS, and other parameters are elaborately tuned to produce their fastest convergence.

\begin{figure}[t]
\begin{center}
\includegraphics[width=\figurewidth]{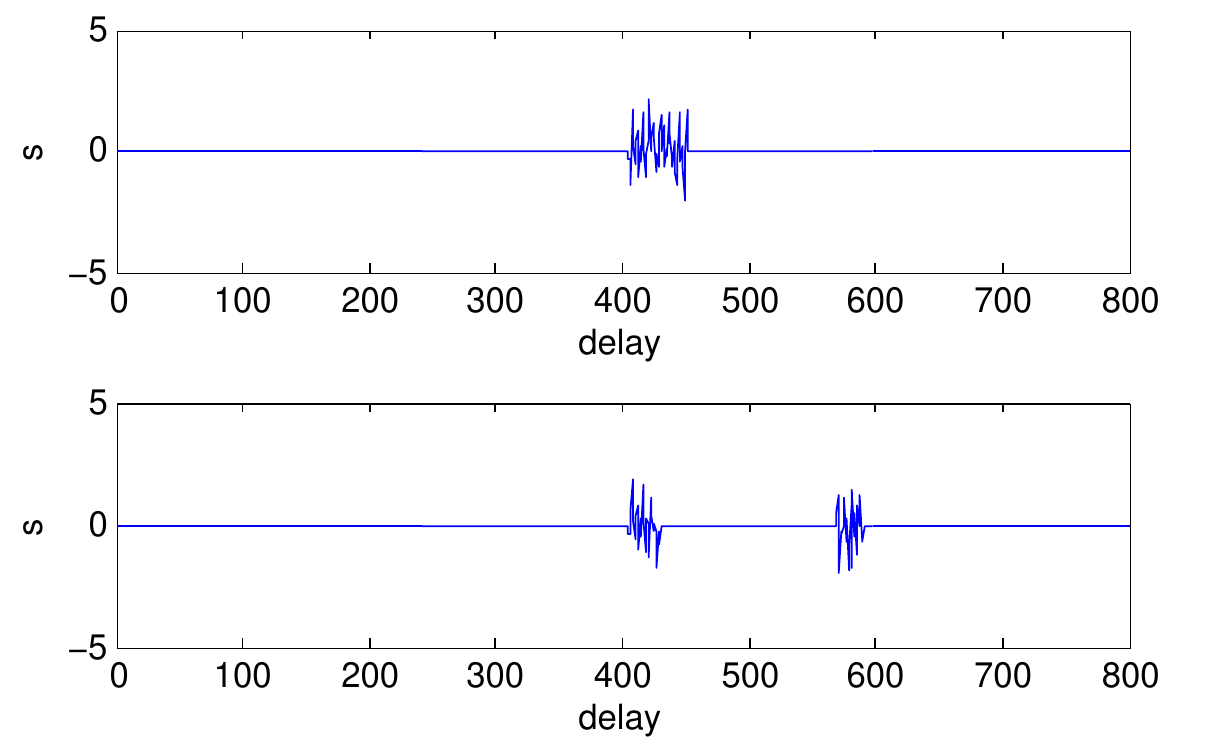}
\caption{The block-sparse systems tested in the first experiment.}
\label{fig:traditionalchannels}
\end{center}
\end{figure}

\begin{figure}[t]
\begin{center}
\includegraphics[width=\figurewidth]{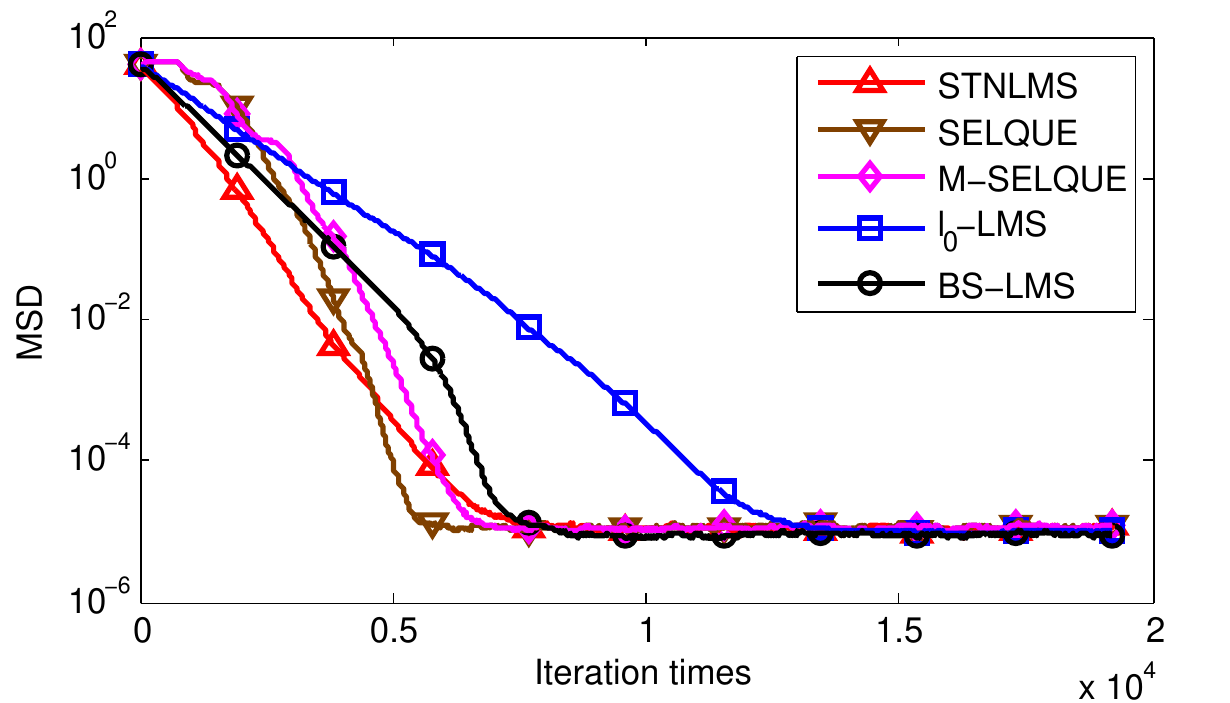}
\includegraphics[width=\figurewidth]{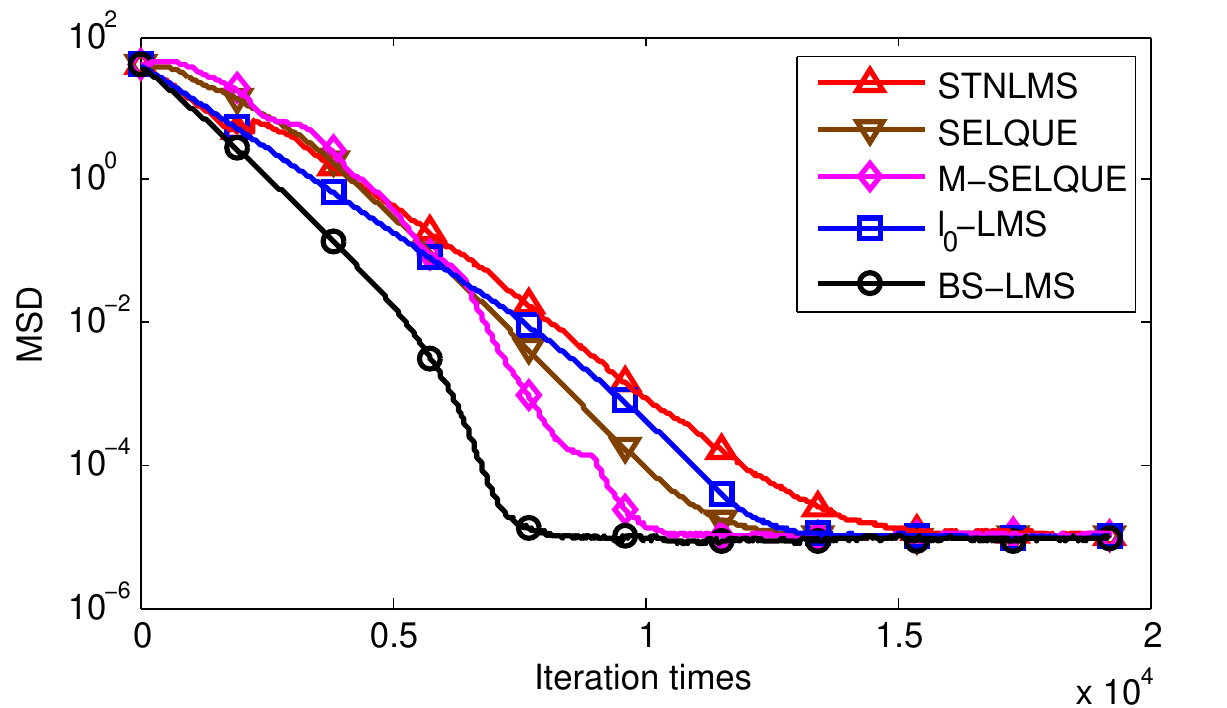}
\caption{The learning curves of the proposed algorithm and the references when identifying the two block-sparse systems displayed in Fig. \ref{fig:traditionalchannels}, where (top) and (bottom) correspond to the single-cluster and multi-cluster system respectively. }
\label{fig:MSD_traditionalchannels}
\end{center}
\end{figure}

According to Fig. \ref{fig:MSD_traditionalchannels}, when there is only one cluster in the system response, the convergence performance of BS-LMS is still lightly inferior to other block algorithms. SELQUE converges the fastest among all the algorithms. Nonetheless, when there exists two clusters, BS-LMS shows its advantage. Because BS-LMS need not detect the active regions, its performance is nearly not affected by the number of clusters. The convergence rates of STNLMS and SELQUE deteriorate significantly, because all of the active regions, along with flat delays between the two clusters, are considered as a long active region. Although M-SELQUE can obviate this problem, its convergence behavior still becomes worse when the unknown system has multi-clusters.

In the second experiment, BS-LMS and the proposed M-G model are tested. The simulation results are plotted in Fig. \ref{fig:MSD_instant_block}, where (a), (b), and (c) correspond to the unknown systems plotted in the diagonal subfigures of Fig. \ref{fig:coeff_markov}, respectively,  from left-top to right-bottom. For BS-LMS, the parameters $\mu$, $\kappa$, and $P$ are set as $0.6/L$, $3.90{\rm e} \textnormal{-}7$, and $3$ for (a), $1/L$, $1.07{\rm e}\textnormal{-}6$, and $4$ for (b), $1/L$, $1.60{\rm e}\textnormal{-}6$, and $5$ for (c). For the reference algorithms, the step-size and other parameters are properly adjusted to get their fastest convergence rate and equal steady-state MSD with BS-LMS.

According to the simulation results, BS-LMS and $l_0$-LMS are always among the best when identifying various block-sparse systems, which are generated by the proposed M-G model with various parameter sets. BS-LMS converges faster than $l_0$-LMS, because the former utilizes the block-sparsity prior. Among all the algorithms based on active region detection, M-SELQUE behaves the best but still converges slower than $l_0$-LMS, because it takes more iterations to identify the locations of nonzero coefficients and thus reduces the convergence rate when there are more and dispersed clusters, which is highly likely to be produced by utilizing the M-G model. SELQUE and STNLMS get the worst performance because they are not suitable to the multi-cluster system, which has been demonstrated in the first experiment. Above all, we can conclude that BS-LMS has a superior robustness than all reference algorithms, especially in the scenarios of multiple-scattered-cluster sparse systems.

\subsection{On the Theoretical Results}

In the third experiment, the steady-state performance of BS-LMS of different group partition size $P$ with respect to $\kappa$ is tested. The unknown system response is shown in the center of Fig. \ref{fig:coeff_markov}. $P$ is chosen as 1, 5, 10, and 20, respectively. For each $P$, $\kappa$ varies from $10^{-9}$ to $10^{-5}$, and $\mu=0.8/L$.

\begin{figure}[t]
\begin{center}
\includegraphics[width=\figurewidth]{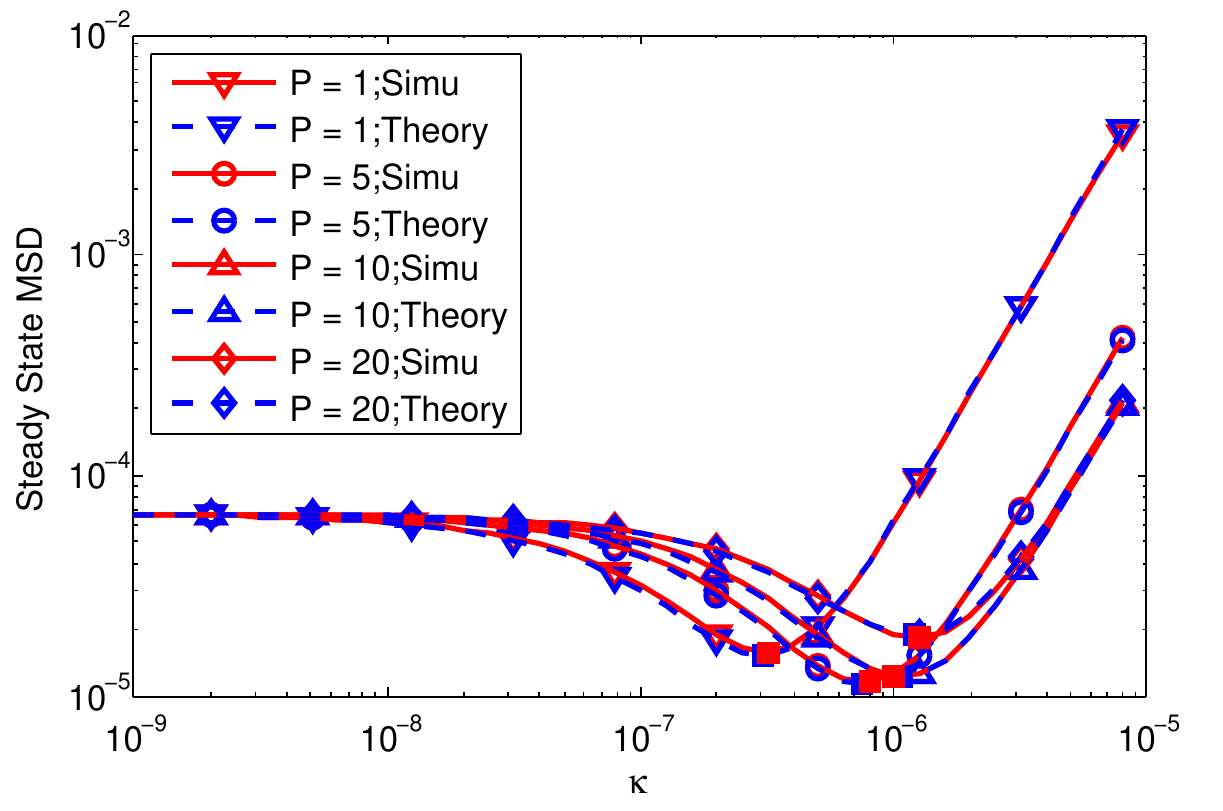}
\caption{Steady-state MSD of BS-LMS of different group partition size with respect to $\kappa$. The solid square denotes the theoretical $\kappa_{\textrm{opt}}$.}
\label{fig:MSD_steady_kappa}
\end{center}
\end{figure}

Referring to Fig. \ref{fig:MSD_steady_kappa}, we can see that the theoretical steady-state MSD of BS-LMS agrees well with the simulation result. For every group partition size, as $\kappa$ increases from $10^{-9}$, the steady-state MSD decreases at first, which means that proper GZA is useful to reduce the amplitudes of coefficients in $\mathcal{C}_0$. However, when $\kappa$ continues to increase, more intense GZA enhances the bias of coefficients in $\mathcal{C}_{\rm S}$. For different group partition size, the minimum steady-state MSD and its corresponding optimal $\kappa$ varies. One may recognize that the simulation result of the optimal $\kappa$ tallies with theoretical $\kappa_{\textrm{opt}}$ very well.

In the fourth experiment, for given M-G model, the effect of group partition size on the average steady-state MSD is investigated. The model parameter set $(p_1, p_2)$ is chosen as $(0.98, 0.82), (0.99, 0.91)$, and $(0.995, 0.955)$, respectively. Three typical system responses are plotted in the diagonal subfigures of Fig. \ref{fig:coeff_markov}. $P$ varies from 1 to 50. For each $P$, $\kappa$ is chosen as the theoretical optimal, and $\mu=0.4/L$.

\begin{figure}[t]
\begin{center}
\includegraphics[width=\figurewidth]{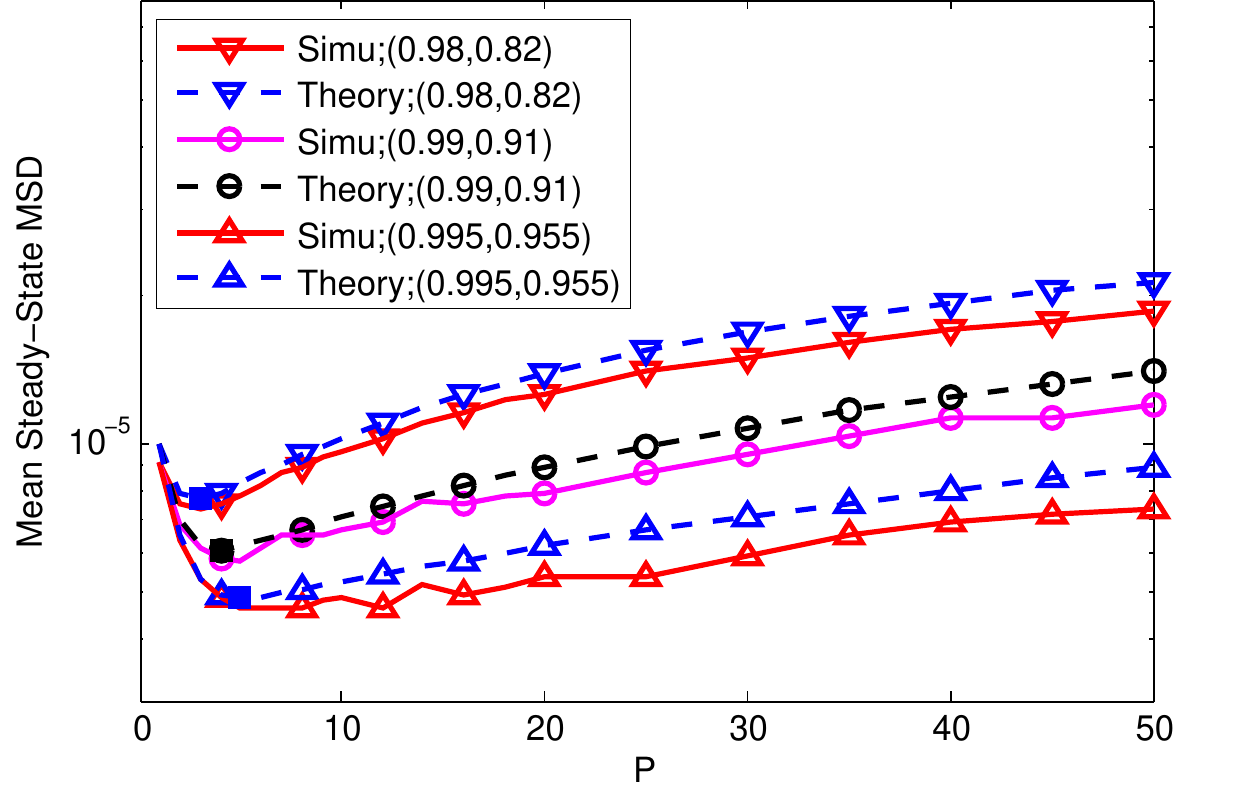}
\caption{Average steady-state MSD of BS-LMS with respect to different group partition size for given M-G model. $\kappa$ is chosen as the theoretical optimal and the solid square denotes $P_{\textrm{opt}}$.}
\label{fig:Bopt_p1_p2}
\end{center}
\end{figure}

Please refer to Fig. \ref{fig:Bopt_p1_p2} for the result. According to Fig. \ref{fig:coeff_markov}, with the growth of $p_1$ and $p_2$, the average block size of nonzero coefficients increase from $5.6$, $11.1$, to $22.2$.
As a consequence, the optimal group partition size $P_{\rm{opt}}$ also increases from $3$, $4$, to $5$. One may find that for all block-sparse systems, when $P$ initially increases from $1$, the minimum MSD decreases quickly at first, which means that treating nonzero coefficients in groups really improves the identification of block-sparse system. However, the minimum MSD increases after $P$ exceeds its optimum, which is much smaller than the average block size, and becomes larger than that of $P=1$, which demonstrates the severe consequence of border effect. The above results accord well with the intuition. Furthermore, simulation results tally with analytical values, especially when $P$ is small.

In the last experiment, for given M-G model, the theoretical transient behavior of BS-LMS with the optimal group partition size is verified by simulation and compared with that of $l_0$-LMS. The model parameter set $(p_1, p_2)$ is chosen as $(0.99, 0.91)$. A typical unknown system response is shown in the center of Fig. \ref{fig:coeff_markov}. The SNR are selected as 40dB and 20dB to test the performance in various noisy scenarios. For BS-LMS and $l_0$-LMS, $\mu$ is chosen as $0.637/L$ and $0.4/L$, respectively, to make their average steady-state MSDs equal. For both algorithms, $P$ and $\kappa$ are chosen as their corresponding optimal values.

\begin{figure}[t]
\begin{center}
\includegraphics[width=\figurewidth]{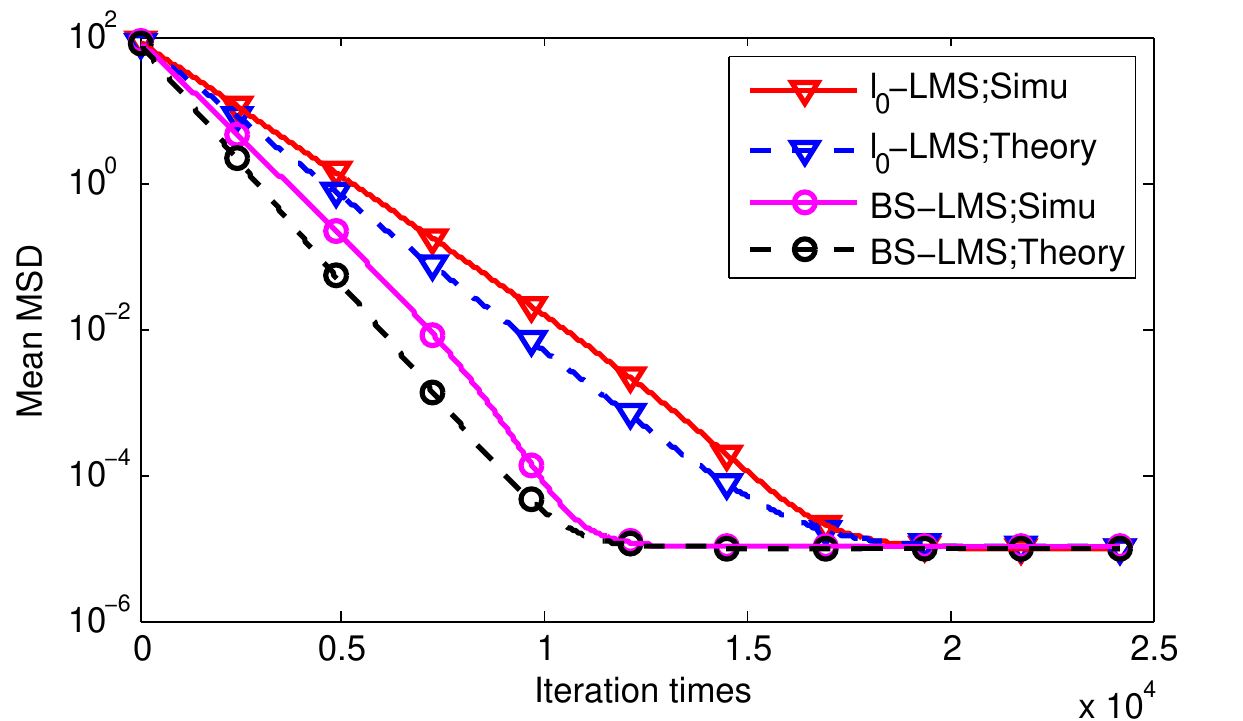}
\caption{Average learning curves of BS-LMS and $l_0$-LMS with given M-G model ($p_1=0.99$ and $p_2=0.91$). $P$ and $\kappa$ are chosen as the optimal. The SNR is 40dB. }
\label{fig:MSD_instant_both}
\end{center}
\end{figure}

\begin{figure}[t]
\begin{center}
\includegraphics[width=\figurewidth]{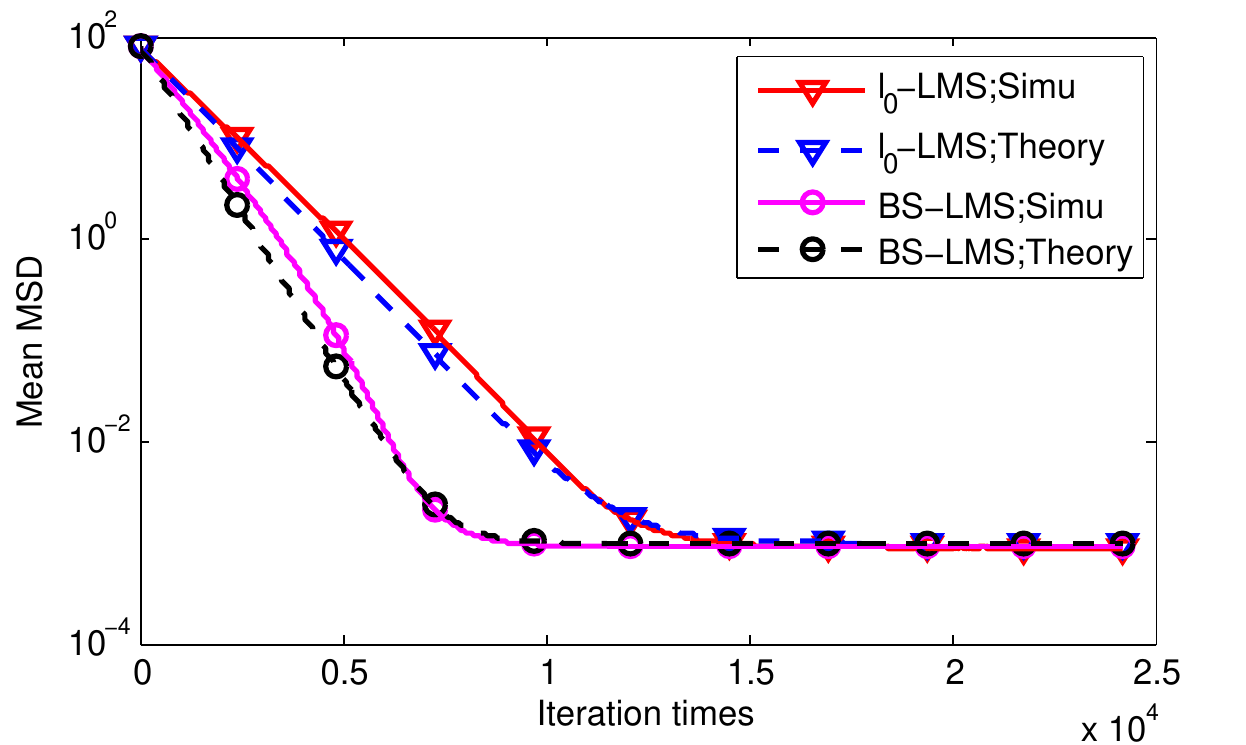}
\caption{Average learning curves of BS-LMS and $l_0$-LMS with given M-G model ($p_1=0.99$ and $p_2=0.91$). $P$ and $\kappa$ are chosen as the optimal. The SNR is 20dB. }
\label{fig:MSD_instant_both_smallSNR}
\end{center}
\end{figure}

Please refer to Fig. \ref{fig:MSD_instant_both} and Fig. \ref{fig:MSD_instant_both_smallSNR}. One may readily see that the convergence rate of BS-LMS is always faster than that of $l_0$-LMS. Furthermore, the theoretical analysis of transient behavior accords with simulation in a tolerable error, which origins mainly from large step-size and the independence assumption.

\section{Conclusion}
\label{sec:conclusion}

In order to improve the performance of block-sparse system identification, a new algorithm based on $l_0$-LMS is proposed in this paper by changing $l_0$ norm to mixed $l_{2, 0}$ norm with equal group partition sizes in the cost function. Also, a M-G model is put forward to describe the block-sparse system. Furthermore, the theoretical analysis on performance of BS-LMS compared to $l_0$-LMS is presented based on the expressions of mean square misalignment, which shows that BS-LMS is better than $l_0$-LMS theoretically. Finally, simulations are designed to verify the theoretical results and confirm superior performance of our proposed algorithm.

\section{Appendix}

\subsection{Expressions of Constants}
\label{append:EoC}

All through this paper, we have
\begin{align}
&\Delta_L \triangleq 2-(L+2)\mu\sigma_x^2, & &\Delta_Q \triangleq 2-(Q+2)\mu\sigma_x^2,\label{def:DeltaL}\\
&\Delta_0 \triangleq 1-\mu\sigma_x^2, & &\Delta_0' \triangleq 2-\mu\sigma_x^2,\label{def:Delta0}
\end{align}
\begin{align}
G({\bf s}) \triangleq & \langle {\bf g}({\bf s}), {\bf g}({\bf s}) \rangle = \sum_{k \in {\mathcal C}_{\rm S}} g_k^2({\bf s}), \label{def:Gammaprime}\\
G'({\bf s}) \triangleq & \langle {\bf s}, {\bf g}({\bf s}) \rangle = \sum_{k \in {\mathcal C}_{\rm S}} s_k g_k ({\bf s}).
\end{align}

In Lemma \ref{lem:MSD_kappa}, the constants $\left\{ \beta_i \right\}$ are
\begin{align}
\!\!\beta_0 \!\! \triangleq & \mu \sigma_x^2 \Delta_0 ' \Delta_L G({\bf s}) \!\! + \!\! 4\alpha^2 \Delta_Q \!\! \left(\frac{\mu \sigma_x^2 \Delta_L}{P} \!\! + \!\! \frac{\Delta_0 \Delta_Q \theta^2(P)}{\pi}\right)\!\!, \label{def:beta_0}\\
\!\!\beta_1 \!\! \triangleq & \frac{ \Delta_0 ' G({\bf s}) \! + \! 4(L \! - \! Q) \alpha^2 \!\! \left(\frac{\mu \sigma_x^2}{P} \! + \! \frac{2\Delta_0 \Delta_Q \theta^2(P)}{\pi\Delta_L}\right )}
{\mu^2 \sigma_x^4 \Delta_L}, \label{def:beta_1}\\
\!\!\beta_2 \!\! \triangleq &
\frac{4 \alpha \left(L - Q \right) \theta(P)}{\mu^2 \sigma_x^4 \Delta_L^2}\sqrt{\frac{\Delta_0 \beta_0}{\pi}}, \label{def:beta_2}\\
\!\!\beta_3 \!\! \triangleq & {2 \mu^3 \sigma_x^4 \sigma_v^2 \Delta_0 \Delta_L}/{\beta_0},\label{def:beta_3}
\end{align}
where $\theta(P)$ is defined as
\begin{equation}
\label{eq:thetaP}
\theta(P) \triangleq
\begin{cases}
\myfrac{\left[ (({P-1})/{2})!\right]^2 2^{P-1}}{P!}, & P\ \textrm{is}\ \textrm{odd}; \\
\myfrac{(P - 1)! \pi}{({P}/{2})! ({P}/{2} - 1)! 2^P}, & P\ \textrm{is}\ \textrm{even}.
\end{cases}
\end{equation}

In the proof of Lemma \ref{lem:MSD_instant}, ${\bf A} = \{ a_{ij}\}$ is defined as
\begin{equation}
\label{def:A}
{\bf A} \!\! \triangleq \!\!
\left[\!\!\begin{array}{ll}
1 - \mu \sigma_x^2 \Delta_L & -\theta(P) \sqrt{\frac{8}{\pi}} \frac{\kappa \alpha \Delta_0}{\omega} \\
\left( L - Q\right) \mu^2 \sigma_x^4 & 1 - 2 \mu \sigma_x^2 \Delta_0 - \theta(P) \sqrt{\frac{8}{\pi}} \frac{\kappa \alpha \Delta_0}{\omega}
\end{array}\!\!\right]
\end{equation}
and ${\bf b}_n$ is also used in the derivation (please refer to \cite{Su2}),
\begin{equation}
	{\bf b}_n \triangleq [b_{0,n}, b_{1,n}]^{\rm T},\label{def:b_n}
\end{equation}
where
\begin{align}
b_{0,n} &\triangleq
L \mu^2 \sigma_x^2 \sigma_v^2 \!+\! \left( L \!-\! Q \right) \!\left( \frac{4 \alpha^2 \kappa^2}{P} \!-\! \theta(P) \sqrt{\frac{8}{\pi}} \kappa \Delta_0 \alpha \omega\right) \nonumber\\
&\quad+ \frac{\kappa^2 \left( \Delta'_0 - 2 \Delta_0^{n+1} \right)}{\mu \sigma_x^2} G({\bf s}) - 2 \kappa \Delta_0^{n+1} G'({\bf s}),\label{def:b_0n} \\
b_{1,n} &\triangleq\left( L\! -\! Q \right) \left( \mu^2 \sigma_x^2 \sigma_v^2 \!+\! \frac{4 \alpha^2 \kappa^2}{P} - \theta(P) \sqrt{\frac{8}{\pi}} \kappa \Delta_0 \alpha \omega \right).\label{def:b_1n}
\end{align}
Please notice that in \eqref{def:A}, \eqref{def:b_0n}, and \eqref{def:b_1n}, $\omega$ is the solution of
\begin{align}
2\mu\sigma_x^2\Delta_0\Delta_L\omega^2 &+ \frac{8\kappa\alpha\theta(P)\Delta_0\Delta_Q}{\sqrt{2\pi}}\omega - 2\mu^2\sigma_x^2\sigma_v^2\Delta_0\nonumber\\
&-\kappa^2\left(\frac{4\alpha^2\Delta_Q}{P} +G({\bf s})\Delta'_0\right) = 0.\label{eq:origomegaequation}
\end{align}

In Lemma \ref{lem:MSD_instant}, the constants $\lambda_3$ and $c_3$ are
\begin{align}
\label{def:lambda_3}
\lambda_3 \triangleq & \Delta_0, \\
c_3 \triangleq & - \frac{2 \kappa \Delta_0}{\mu \sigma_x^2} \frac{\mu \sigma_x^2 - 2 \mu^2 \sigma_x^4 + \theta(P) \sqrt{\frac{8}{\pi}} \frac{\kappa \alpha \Delta_0}{\omega}}{{\rm det} \left( \lambda_3 {\bf I} - {\bf A}\right)} \nonumber \\
& \cdot \left( \kappa G({\bf s}) + \mu \sigma_x^2 G'({\bf s})\right). \label{def:c_3}
\end{align}

\subsection{Proof of Property \ref{prp:bordereffect}}
\label{append:bordereffect}
When we choose a group partition size $P$ to divide the unknown system coefficients, $L/P$ groups are obtained. For simplicity, we consider every group independently. Here we denote the number of nonzero coefficients in a group as a random variable $M$ and the probability that $M$ equals $m, 0\le m\le L,$ as ${\rm P}\left\{ M=m \right\}$. Based on the definition of transfer matrix and the fact that the Markov process is in steady-state distribution, we can get the solution of $m=0$ as
$$
{\rm P} \left\{M\!=\!0\right\} = \frac{1 - p_2}{2 - p_1 - p_2} p_1^{P - 1}.
$$
Utilizing Property \ref{prp:sparsity}, we then have
\begin{equation}
\overline{Q} = \frac{L}{P} (P \cdot {\rm P} \left\{ M > 0 \right\} ) = L \left( 1 - (1 - \overline S) p_1^{P-1}\right).
\end{equation}

\subsection{Proof of Theorem \ref{thm:MSD_B}}
\label{append:proofofthmMSDB}

\begin{proof}
In order to get the mean steady-state MSD of BS-LMS, we need to simplify the result in Lemma \ref{lem:MSD_kappa1}. Before giving out approximations, we will prove a useful inequality
\begin{equation}
G({\bf s}) < 4 \alpha^2 Q / P.
\label{eq:ineq_Gs}
\end{equation}
For small group coefficients, according to \eqref{eq:gtfunction}, we know that
\begin{equation}
\left| g_k ({\bf s}) \right| < \frac{\textstyle 2 \alpha {\left| s_{k}\right|}}{\textstyle \left\| {\bf s}_{[ \lceil k\!/\!P \rceil]}\right\|_2} .
\end{equation}
As the number of small groups is no larger than $Q/P$ and the sum of squares of $g_k ({\bf s})$, which belongs to the same small group, is less than $4 \alpha^2$, we get the inequality \eqref{eq:ineq_Gs}.

Then we will present some approximations. Utilizing $\Delta_0 \approx 1$, $\Delta'_0 \approx 2$, which are derived from Assumption \ref{assump:smallP}), and \eqref{eq:ineq_Gs} in \eqref{def:beta_0}, \eqref{def:beta_1}, \eqref{def:beta_2}, and \eqref{def:beta_3}, we have
\begin{align}
\beta_0 \approx & {4 \alpha^2 \Delta_0 \Delta_Q^2} \theta^2(P) /{\pi}\approx {4 \alpha^2 \Delta_Q^2} \theta^2(P)/{\pi}, \label{def:beta_0_approx}\\
\beta_1 \approx & \frac{1}{\pi \mu^2 \sigma_x^4 \Delta_L^2} {8 (L - Q) \alpha^2 \Delta_0 \Delta_Q}\theta^2(P) \nonumber \\
\approx & \frac{1}{\pi \mu^2 \sigma_x^4 \Delta_L^2} {8 (L - Q) \alpha^2 \Delta_Q} \theta^2(P), \label{def:beta_1_approx}\\
\sqrt{\beta_1^2 - \beta_2^2} \approx & \frac{1}{{\mu^2 \sigma_x^4 \Delta_L^2}} \bigg[ \left(  8 (L \!\! - \!\! Q) \alpha^2 \Delta_Q \theta^2(P) / \pi \! + \! 2 \Delta_L G({\bf s})\right)^2 \nonumber \\
- 16 \alpha^2 (L \! - \! & Q)^2  \theta^2(P)  \left( 4 \alpha^2 \Delta_Q^2 \theta^2(P) / \pi \!\! + \!\! 2 \mu \sigma_x^2 \Delta_L G({\bf s}) \right)/\pi \bigg]^{\frac{1}{2}} \nonumber \\
\approx & \frac{4\alpha\theta(P)\sqrt{2 (L - Q) G({\bf s}) / \pi}}{\mu^2 \sigma_x^4 \Delta_L},\label{def:beta_1_2_approx}\\
\beta_3 & \approx \frac{\pi\mu^3\sigma_x^4\sigma_v^2\Delta_L}{2\alpha^2\Delta_Q^2\theta^2(P)}.\label{def:beta_3_approx}
\end{align}

Utilizing \eqref{def:beta_1_approx}, \eqref{def:beta_1_2_approx}, and \eqref{def:beta_3_approx} in \eqref{eq:MSD_kappa}, one achieves a temporary result that
\begin{equation}\label{eq:minMSD_approx1}
D_{\infty}^{\rm{min}} \approx \frac{\mu \sigma_v^2}{\Delta_L}\left(L - \frac{2 (L - Q)}{\Delta_Q}\right) + \frac{\mu \sigma_v^2\sqrt{2 \pi(L - Q) G({\bf s}) }}{\alpha \theta(P) {\Delta_{Q}^2}}.
\end{equation}
The first item in the RHS of \eqref{eq:minMSD_approx1} could be further approximated by adopting Assumption \ref{assump:smallP}) and we finally arrive
\begin{equation}\label{eq:minMSD_approx2}
D_{\infty}^{\rm{min}} \approx \frac{\mu\sigma_v^2}{\Delta_Q}\left(Q + \frac{\sqrt{2 \pi(L - Q) G({\bf s}) }}{\alpha \theta(P) {\Delta_{Q}}}\right).
\end{equation}

For the sake of mathematical tractability, we replace $Q$ and $G({\bf s})$ in the above equation with their means, respectively, to yield the final average steady-state MSD of \eqref{eq:MSD_B}. One may accept that there is no other choice because the formula of \eqref{eq:minMSD_approx2} is highly sophisticated. However, simulation result will verify that the approximation produces acceptable errors.

What remains in finishing the proof is to derive the expression of $\overline{Q}$ and $\overline{G({\bf s})}$. The former could be gotten based on Property \ref{prp:bordereffect} of M-G model, and we define
\begin{equation}\label{def:mean_Delta_Q}
\Delta_{\overline{Q}} \triangleq 2 - (\overline{Q} + 2) \mu \sigma_x^2.
\end{equation}
In the following, we will conduct the derivations of $\overline{G({\bf s})}$.

Assuming that there are $m$ nonzero unknown coefficients in a certain small group, we have
\begin{equation}\label{eq:niid}
\sum_{k=1}^m g_k^2({\bf s}) = 4 \alpha^4 \sum_{k=1}^m s_k^2 - 8 \alpha^3 \left(\sum_{k=1}^m s_k^2\right)^{\frac12} + 4 \alpha^2.
\end{equation}
We denote the mean of the LHS of (\ref{eq:niid}) as $F_{\alpha}(m)$. According to the M-G model that the nonzero coefficients follow \emph{i.i.d.} Gaussian distribution and the property of $\chi^2$ distribution, $F_{\alpha}(m)$ is obtained as
\begin{align}
F_{\alpha}(m) = & \frac{4\alpha^2}{\Gamma(m/2)} \bigg[2 \alpha^2 \sigma_s^2 \gamma \left(\frac{m+2}{2},\frac{1}{2 \alpha^2}\right) \! + \! \gamma \left(\frac{m}{2},\frac{1}{2 \alpha^2}\right) \nonumber \\ & - 2 \sqrt{2} \alpha \sigma_s \gamma \left(\frac{m+1}{2},\frac{1}{2 \alpha^2}\right)\bigg], \label{def:Falphan}
\end{align}
where $\Gamma(\cdot)$ and $\gamma(\cdot,\cdot)$ denote the ordinary gamma function and the lower incomplete function, respectively.
Then we know that
\begin{equation}\label{def:Gamma1_mean}
\overline{G({\bf s})} = \frac{L}{P} \sum_{m=1}^P {\rm P}\left\{ M=m\right\}F_{\alpha}(m).
\end{equation}
where
${\rm P}\{M = m\}$ follows the definition in Appendix \ref{append:bordereffect} and is further solved as
\begin{align}
&{\rm P} \left\{M\!=\!m\right\} \nonumber\\
&\!\!=\!\!\begin{cases}
\myfrac{1 - p_2}{2 - p_1 - p_2} p_1^{P - 1}, & m \!\! = \!\!0; \\
\left(1 - \myfrac{(1 - p_2) p_1^{P-1} + (1 - p_1) p_2^{P-1}}{2 - p_1 - p_2} \right)  \\
\cdot \myfrac{1 - \myfrac{p_2}{p_1}}{1 - (\myfrac{p_2}{p_1})^{P - 1}} (\myfrac{p_2}{p_1})^{m - 1},
 & 0 \!\! < \!\! m \!\! < \!\! P;\\
\myfrac{1 - p_1}{2 - p_1 - p_2} p_2^{P - 1}, & m \!\! = \!\!P.
\end{cases} \label{def:Pxk}
\end{align}
Then the AMS-MSD is finally achieved.

Considering its sophisticated shape of expression \eqref{eq:MSD_B}, we prefer to numerically solve the optimal parameter by \eqref{eq:Popt}. Thus the proof of Theorem \ref{thm:MSD_B} is completed.
\end{proof}

\subsection{Proof of Theorem \ref{thm:meanMSD_instant}}
\label{append:proofoflemmaMSDinstant}

\begin{proof}
From Lemma \ref{lem:MSD_instant}, we have already gotten the close forms of instantaneous MSD, $c_3$, and $\lambda_3$. Here we will show how to get the approximate close forms of $\lambda_1, \lambda_2, c_1$, and $c_2$. Note that all expressions of $\{c_i, \lambda_i\}$ are attached by $(\cdot)^\prime$ here to distinguish them from those in Lemma \ref{lem:MSD_instant}.

We will first present the sketch of our proof. Based on Lemma \ref{lem:MSD_instant}, we know that $\lambda_1$ and $\lambda_2$ are the eigenvalues of matrix ${\bf A}$, which is defined by \eqref{def:A}. Utilizing Assumption \ref{assump:smallP}), which derives $\Delta_0\approx1$ and $\Delta_Q\approx2$, in \eqref{def:A}, we could simplify $\det(\lambda{\bf I}-{\bf A})=0$ by
\begin{align}
(\lambda')^2 &- \left( 2 - 2\mu\sigma_x^2 - \mu C(\mu,\omega) \right) \lambda' \nonumber\\
&+ 1 - 2\mu\sigma_x^2 - \mu C(\mu,\omega) +2\mu^2\sigma_x^2 C(\mu,\omega)  = 0,\label{eq:quadraticequation}
\end{align}
where
\begin{equation}\label{eq:defC}
	C(\mu,\omega) \triangleq \sigma_x^2\Delta_L + \sqrt{\frac{8}{\pi}}  \frac{\kappa \alpha}{\mu\omega }\theta(P).
\end{equation}
We then solve the quadratic equation of \eqref{eq:quadraticequation} and get the close forms of
$$
\lambda'_1 \triangleq  1 - 2\mu \sigma_x^2, \qquad
\lambda'_2 \triangleq  1 - \mu C(\mu,\omega).
$$
As a consequence, $c'_1$ and $c'_2$ are obtained by satisfying initial values,
\begin{align}
c'_{1,2} \triangleq & \frac{ (1 - \mu \sigma_x^2 \Delta_L) \left\| {\bf s}\right\|_2^2 + b_{0,0} - D_{\infty} - c'_3 \lambda'_3}{\lambda'_{1,2} - \lambda'_{2,1}} \nonumber \\
& - \frac{\lambda'_{2,1} \left({\left\| {\bf s}\right\|_2^2} - c'_3 - D_{\infty}\right)}{\lambda'_{1,2} - \lambda'_{2,1}},\label{def:c_12}
\end{align}
where $\lambda'_3$ and $c'_3$ are defined in \eqref{def:lambda_3} and \eqref{def:c_3}, respectively.
If we could prove that
\begin{equation}\label{eq:defapproxC}
C(\mu,\omega) \approx 2 \sigma_x^2\alpha \theta(P) \sqrt{\frac{2 \left( L - Q\right) }{\pi G({\bf s})}},
\end{equation}
Thoerem \ref{thm:meanMSD_instant} is ready to be proved.

Next we will prove \eqref{eq:defapproxC}. By taking the approximation of $\Delta_0 \approx 1, \Delta'_0 \approx \Delta_Q \approx 2$ and  utilizing the optimal $\kappa$ in \eqref{eq:origomegaequation}, we get
\begin{equation}
\mu\sigma_x^2\Delta_L\omega^2 + \frac{8\kappa_{\rm opt}\alpha\theta(P)\omega}{\sqrt{2\pi}} - \mu^2\sigma_x^2\sigma_v^2
-\kappa_{\rm opt}^2\!\left(\!\frac{4\alpha^2}{P} +G({\bf s})\!\right)\! = \!0.\label{eq:omegaequation}
\end{equation}
Solving \eqref{eq:omegaequation} to get the closed form of $\omega$ and inserting it into $\kappa_{\rm opt} / ( \mu\omega)$, we get
\begin{equation}
\frac{\kappa_{\rm opt}}{\mu\omega } = \frac{\sigma_x^2 \Delta_L}{- \sqrt{\frac{8}{\pi}} \alpha \theta(P)  + \sqrt{t_1  +  t_2 + t_3}},\label{eq:kappaomegamu}
\end{equation}
where
\begin{align}
t_1 &\triangleq 8 \alpha^2 \theta^2(P) / \pi,\label{def:t_1}\\
t_2 &\triangleq \mu^3 \sigma_x^4 \sigma_v^2 \Delta_L / \kappa_{\rm opt}^2
\approx 16\alpha^2\theta^2(P)/(\pi t_4),\label{def:t_2}\\
t_3 &\triangleq \mu\sigma_x^2\Delta_L \left( 4\alpha^2/P + G(s) \right),\\
t_4 &\triangleq \sqrt{\frac{L-Q}{2\pi G({\bf s})}}\frac{4\alpha\theta(P)}{\Delta_L}-1.\label{def:t_4}
\end{align}
Please note that the approximation in \eqref{def:t_2} is produced by utilizing \eqref{def:beta_0_approx}, \eqref{def:beta_1_approx}, \eqref{def:beta_1_2_approx}, \eqref{def:beta_3_approx}, and $\Delta_Q \approx 2$ in \eqref{eq:kappa_optimal}.

If it can be proved that
\begin{equation}
G(s) \ll L \alpha^2 / P,\label{eq:rangeofGs}
\end{equation}
we see that $t_1$ is much larger than $t_3$.
Inserting \eqref{def:t_1} and \eqref{def:t_2} in \eqref{eq:kappaomegamu} and omitting $t_3$, we get
\begin{equation}
\frac{\kappa_{\rm opt}}{\mu\omega } = \frac{\sigma_x^2\Delta_L}{\sqrt{\frac{8}{\pi}}\alpha\theta(P)\left(-1+\sqrt{1+2/t_4}\right)} \approx\sqrt{\frac{\pi}8}\frac{\sigma_x^2\Delta_L t_4}{\alpha\theta(P)},\label{eq:newkappaomegamu}
\end{equation}
where $t_4\gg 1$, which can be derived based on \eqref{def:t_4} and \eqref{eq:rangeofGs}, is adopted in \eqref{eq:newkappaomegamu}.
Utilizing \eqref{def:t_4} and \eqref{eq:newkappaomegamu} in \eqref{eq:defC}, we finally prove \eqref{eq:defapproxC}.

Now we will prove \eqref{eq:rangeofGs}. According to the empirical hypothesis that ${\rm max} (F_{\alpha} (m) / \alpha^2) \sim O(1)$ and two properties that ${\rm P} \left\{ M = m\right\} \ll 1$ with $1 \leq M \leq P$ and that $F_{\alpha}(m)$ decreases dramatically with the increase of $m$ based on the property of $\gamma(\cdot, \cdot)$, we get that $\overline{G(s)} \ll L \alpha^2 / P$. Finally we adopt the method used in the proof of Theorem \ref{thm:MSD_B} by replacing the random variables in Lemma \ref{lem:MSD_instant} by their expectations to produce an average result.

\end{proof}

In Theorem \ref{thm:meanMSD_instant}, $\overline {\left\| {\bf s} \right\|_2^2}$ and $\overline{G'({\bf s})}$ are defined as, respectively,
\begin{align}
\label{def:norms_mean}
\overline {\left\| {\bf s} \right\|_2^2} = & L\overline{S}\sigma_s^2,\\
\label{def:Gamma2_mean}
\overline{G'({\bf s})} = & \frac{L}{P} \sum_{m=1}^{P} {\rm P}\left\{ M = m\right\} F'_{\alpha} \left( m\right),
\end{align}
where $F'_{\alpha}(m)$ is defined as
$$
\frac{1}{\Gamma (m/2)}\!\! \left[ 4 \alpha^2 \sigma_s^2 \gamma \left( \!\frac{m\!+\!2}{2}, \frac{1}{2 \alpha^2}\!\right)
- 2 \sqrt{2} \alpha \sigma_s \gamma \left(\! \frac{m\!+\!1}{2}, \frac{1}{2 \alpha^2}\!\right) \right].
$$

\subsection{Proof of Corollary \ref{cor:MSD_smalllambda}}
\label{append:proofofcorollaryMSD_smalllambda}

\begin{proof}
Because the step-size in BS-LMS with $P_{\rm opt}$ is larger than that in $l_0$-LMS, it is obvious that $\lambda'_1$ and $\lambda'_3$ in BS-LMS with $P_{\rm opt}$ is smaller. In order to compare $\lambda'_2$ between BS-LMS and $l_0$-LMS, we investigate $(L - \overline{Q}) \theta^2(P) / \overline{G({\bf s})}$ in \eqref{eq:lambda_2}.

According to the property of $\gamma(a, x)$, $F'_{\alpha}(n)$ decreases dramatically with the increase of $n$. Thus $\sum_{m=1}^P {\rm P} \left\{ M=m\right\} F'_{\alpha}(m)$ is mainly determined by its first item. We then have
$$
\frac{(L - \overline{Q}) \theta^2(P)}{\overline{G({\bf s})}} \approx \frac{(1 - p_2) \theta^2(P) P f(P)}{(1 - p_2/p_1) F_{\alpha}(1)},
$$
where $f(P)$ is defined as
$$
f(P) \!\! \triangleq \!\!
\begin{cases}
\frac{\textstyle 1 - p_2/p_1}{\textstyle 1 - p_1}, &  P = 1; \\
\frac{\textstyle p_1^{P - 1} - p_2^{P - 1}}{\textstyle(1 - p_2)(1 - p_1^{P - 1}) + (1 - p_1)(1 - p_2^{P - 1})}, & P \geq 2.
\end{cases}
$$

First, we show that $\theta^2(P) P$ is no less than $\theta^2(1) \cdot 1$. When $P$ is odd, $\theta^2(P) P$ can be expressed as
$$
\theta^2(P) P =
\begin{cases}
\myfrac{(P-1)!!(P-1)!!}{P!!(P-2)!!}, &  P \ \textrm{is} \ \textrm{odd} \ \textrm{and} \ P > 1; \\
1, & P=1.
\end{cases}
$$
We see that $\theta^2(P) P$ increases when $P$ becomes larger. Similarly, the same conclusion is reached when $P$ is even. Moreover, we can see that $\theta^2(2) \cdot 2 = \pi^2/8 > 1 = \theta^2(1) \cdot 1$. Therefore  $\theta^2(P) P \geq \theta^2(1) \cdot 1$ is gotten.

Next, we prove $f(P_{\rm opt}) \geq f(1)$ by utilizing the condition that $\overline{Q(P_{\rm opt})}/L = 1 - (1 - p_2)p_1^{P_{\rm opt} - 1}/(2 - p_1 - p_2) \ll 1$. $f(P_{\rm opt}) \geq f(1)$ is equivalent to
\begin{align}
L(P_{\rm opt}) & = \frac{1 - p_1}{1 - p_2} \left( 1 - (p_2/p_1)^{P_{\rm opt}}\right) \nonumber \\
& \geq R(P_{\rm opt}) = \frac{\overline{Q(P_{\rm opt})}/L}{1 - \overline{Q(P_{\rm opt})}/L}(1 - p_2/p_1).
\label{eq:LRP}
\end{align}
If $2 \leq P_{\rm opt} \leq p_1/(p_1 - p_2)$, we just need to prove $L(2) \geq R(2)$ and $L(p_1/(p_1 - p_2)) \geq R(p_1/(p_1 - p_2))$ because $L(P)$ and $R(P)$ are concave and convex respectively. When $P = 2$, we have
$$
L(2) \!\! = \!\! \frac{1 - p_1}{1 - p_2}(1 - \frac{p_2}{p_1})(1 + \frac{p_2}{p_1}) \!\! > \!\! R(2) \!\! = \!\! \frac{1 - p_1}{1 - p_2}(1 - \frac{p_2}{p_1})\frac{2 - p_2}{p_1},
$$
based on that $p_1$ and $p_2$ are close to $1$. When $P = p_1/(p_1 - p_2)$, we take the first-order Taylor expansion of $L(P)$ and $R(P)$. If the Taylor expansion of $L(P)$ is far larger than that of $R(P)$, \eqref{eq:LRP} is valid. we know that
\begin{align}
&L(\frac{p_1}{p_1 \!\! - \!\! p_2}) \!\! \approx \!\! \frac{1 \!\! - \!\! p_1}{1 \!\! - \!\! p_2}(1 \!\! - \!\! \frac{p_2}{p_1}) (1 + \frac{p_2}{p_1 - p_2}) \!\! \nonumber \\
& \gg \!\! R(\frac{p_1}{p_1 \!\! - \!\! p_2}) \!\! \approx \!\! \frac{1 \!\! - \!\! p_1}{1 \!\! - \!\! p_2} (1 \!\! - \!\! \frac{p_2}{p_1}) \left(1 \!\! + \!\! \frac{(2 \!\! - \!\! p_1 \!\! - \!\! p_2) \frac{p_2}{p_1 - p_2}}{1 - \frac{1 - p_1}{p_1 - p_2}p_2}\right),
\end{align}
based on that $p_2/(p_1 - p_2) \gg 1$, $2 - p_1 - p_2 \ll 1$ and $(1 - p_1)/(p_1 - p_2) \ll 1$. Therefore \eqref{eq:LRP} is satisfied when $2 \leq P_{\rm opt} \leq p_1/(p_1 - p_2)$. If $P_{\rm opt} > p_1/(p_1 - p_2)$, we have that
\begin{align}
L(P_{\rm opt}) > & \frac{1 - p_1}{1 - p_2} \left( 1 - (p_2/p_1)^{p_1/(p_1 - p_2)}\right) \nonumber \\
> & \frac{1 - p_1}{1 - p_2} (1 - 1/{\rm e}) \geq \frac{1}{3} (1 - p_2) > \frac{1}{3} (1 - p_2/p_1) \nonumber \\
> & \frac{\overline{Q(P_{\rm opt})}/L}{1 - \overline{Q(P_{\rm opt})}/L}(1 - p_2/p_1) = R(P_{\rm opt}). \nonumber
\end{align}
Thus the conclusion that $f(P_{\rm opt}) \geq f(1)$ is reached.

Above all, it can be seen that $(L - \overline{Q}) \theta^2(P) / \overline{G({\bf s})}$ is larger when $P$ is $P_{\rm{opt}}$. Therefore, $\lambda'_2$ in BS-LMS with optimal $P$ is smaller than that in $l_0$-LMS. Thus the proof of Corollary \ref{cor:MSD_smalllambda} is arrived.

\end{proof}

\begin{figure}[t]
\begin{center}
\includegraphics[width=\figurewidth]{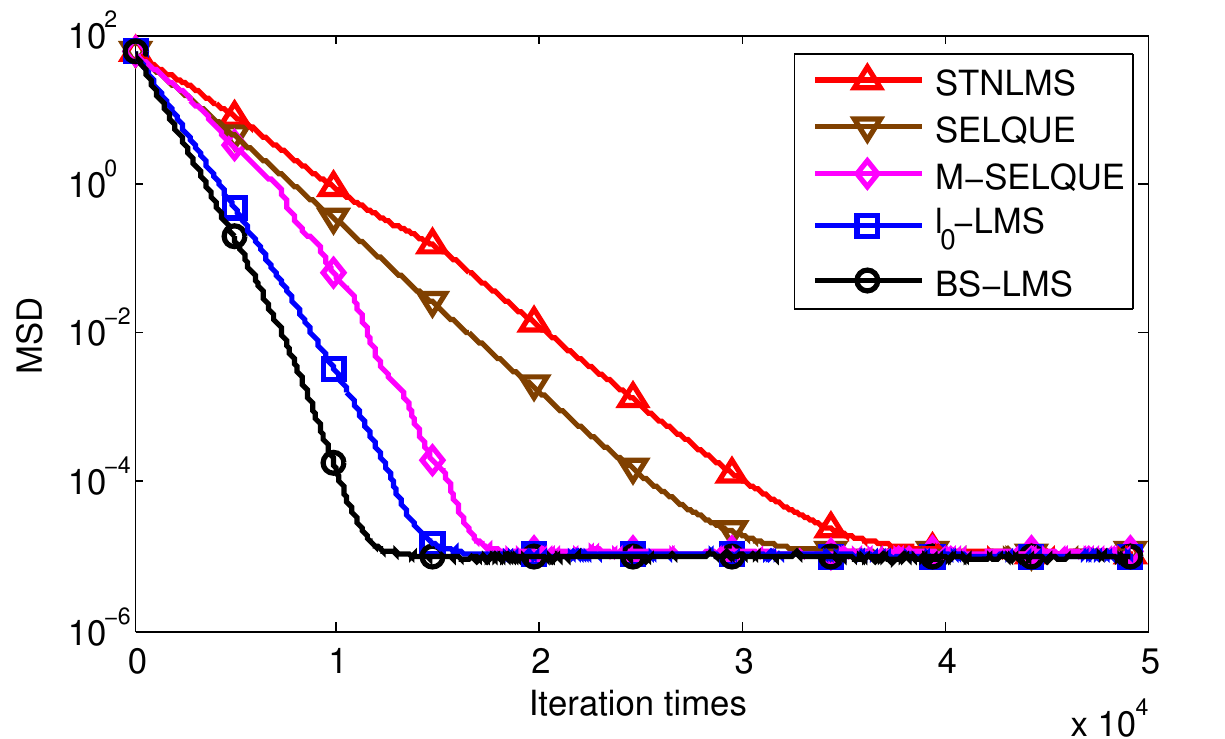}
\includegraphics[width=\figurewidth]{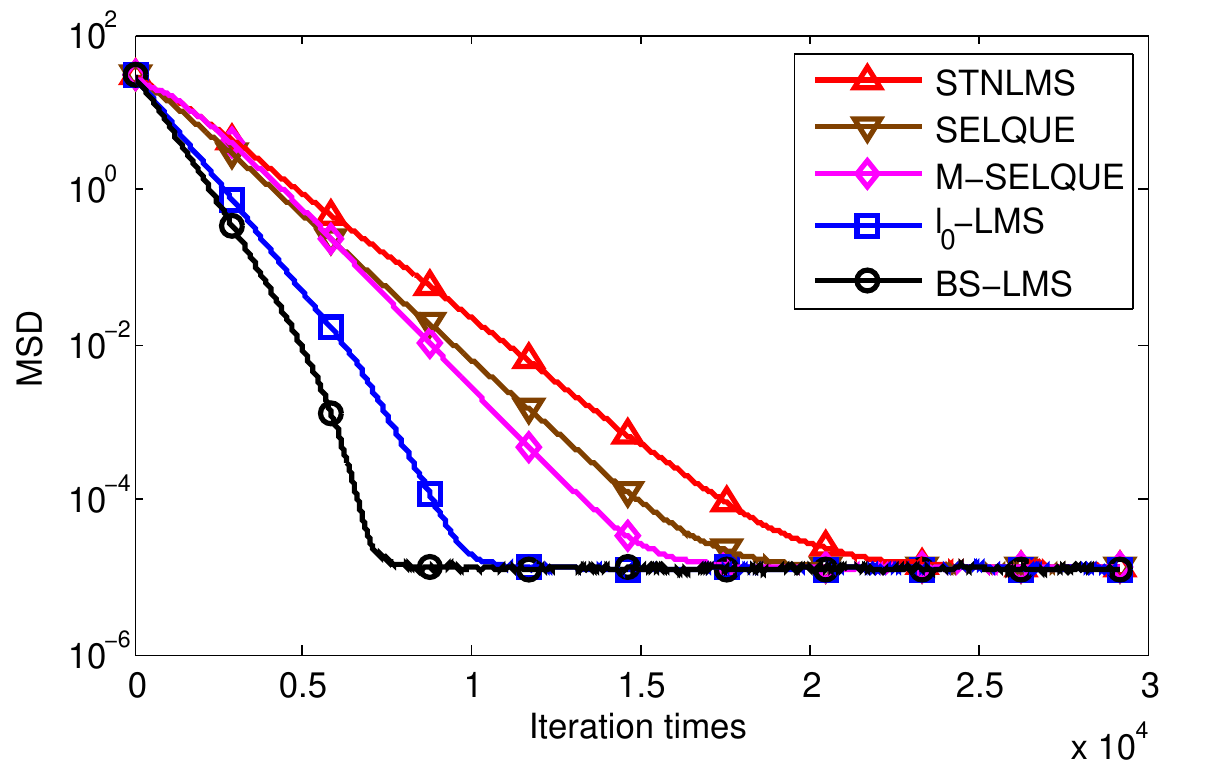}
\includegraphics[width=\figurewidth]{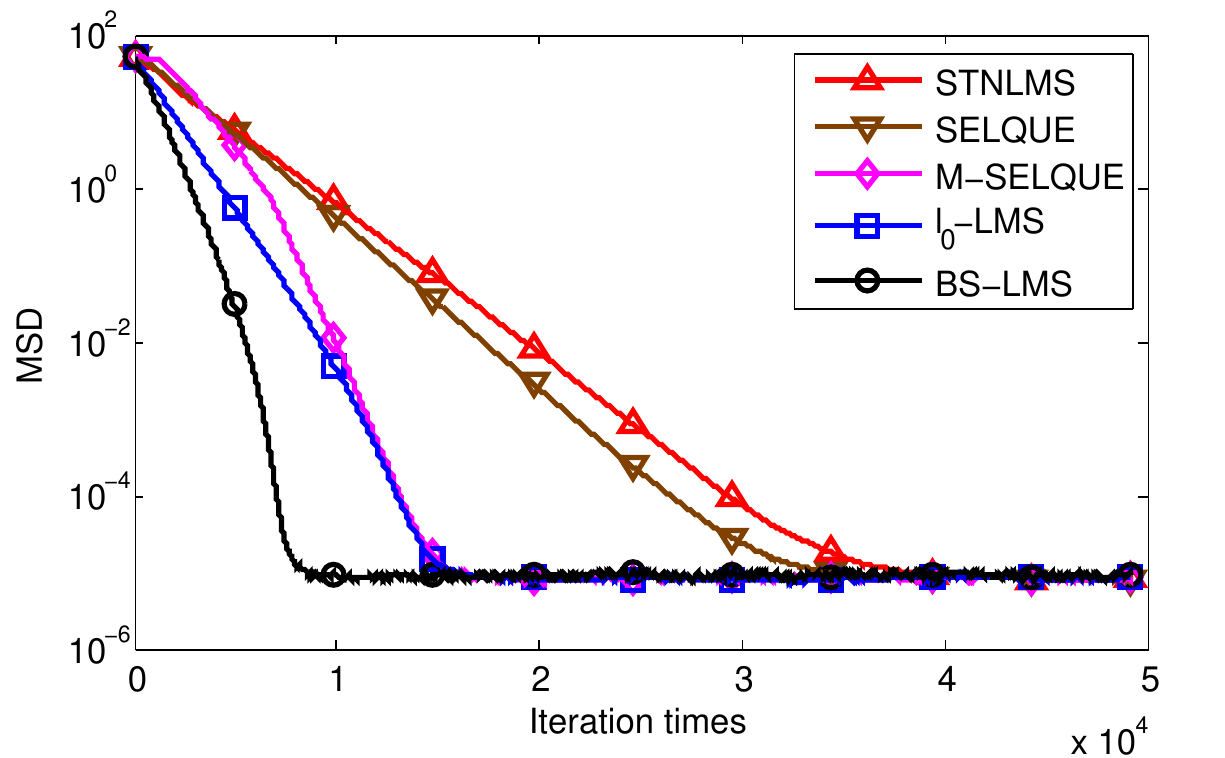}
\caption{The learning curves of the proposed algorithm and the references when identifying three different unknown systems whose impulse response are plotted in the diagonal subfigures of Fig. \ref{fig:coeff_markov}, where (top), (middle), and (bottom) corresponding to, respectively, the left-top, the middle, and the right-bottom. }
\label{fig:MSD_instant_block}
\end{center}
\end{figure}


\footnotesize
\begin{thebibliography}{1}

\bibitem{Haykin}
S.~Haykin, \emph{Adaptive Filter Theory}. Englewood~Cliffs, NJ: Prentice-Hall, 1986.
\bibitem{Schreiber}
W.~F.~Schreiber, ``Advanced television systems for terrestrial broadcasting: Some problems and some proposed solutions,'' \emph{Proc. IEEE}, vol.~83, no.~6, pp.~958-981, Jun. 1995.
\bibitem{Duttweiler}
D.~L.~Duttweiler, ``Proportionate normalized least-mean-squares adaptation in echo cancellers,'' \emph{IEEE Trans. Speech Audio Process.}, vol.~8, no.~5, pp.~508-518, Sep. 2000.
\bibitem{ITU}
``Transmission systems and media, digital systems and networks'', recommendation ITU-T G.168 (2002).
\bibitem{Sugiyama}
A.~Sugiyama, K.~Anzai, H.~Sato, and A.~Hirano, ``Cancellation of multiple echoes by multiple autonomic and distributed echo canceler units,'' \emph{IEICE Trans. Fund.}, vol.~E81-A, no.~11, pp.~2361-2369, Nov. 1998.
\bibitem{Widrow}
B.~Widrow and S.~D.~Stearns, \emph{Adaptive Signal Processing}. Englewood Cliffs, NJ: Prentice-Hall, 1985.
\bibitem{Margo}
J.~H.~Gross, D.~M.~Etter, V.~A.~Margo, and N.~C.~Carlson, ``A block selection adaptive delay filter algorithm for echo cancellation,'' in \emph{Midwest Conf. Circuits Syst.}, Aug. 1992, pp.~895-898.
\bibitem{Margo1}
V.~A.~Margo, D.~M.~Etter, N.~C.~Carlson, and J.~H.~Gross, ``Multiple short-length adaptive filters for time-varying echo cancellation,'' \emph{IEEE ICASSP}, 1993, pp.~I161-I164.
\bibitem{Berg}
M.~Berggren, M.~Borgh, C.~Schuldt, F.~Lindstrom and I.~Claesson, ``Low-Complexity Network echo cancellation approach for systems equipped with external memory,'' \emph{IEEE Trans. Audio, Speech and Language Process.}, vol.~19, no.~8, pp.~2506-2515, Nov. 2011.
\bibitem{Gu}
Y.~Gu, Y.~Chen, and K.~Tang, ``Network echo canceller with active taps stochastic localization,'' \emph{IEEE ISCIT}, pp.~556-559, 2005.
\bibitem{Li}
Y.~Li, Y.~Gu, and K.~Tang, ``Parallel NLMS filters with stochastic active taps and step-sizes for sparse system identification,'' \emph{IEEE ICASSP}, vol.~3, pp.~109-112, Toulouse, France, 2006.
\bibitem{Liu}
X.~Liu, Y.~Li, Y.~Gu, and K.~Tang, ``Enhanced stochastic taps NLMS filter with efficient sparse taps localization,'' \emph{IEEE ICSP}, vol.~4, pp.~16-20, 2006.
\bibitem{Sugiyama2}
A.~Sugiyama, H.~Sato, A.~Hirano, and S.~Ikeda, ``A fast convergence algorithm for adaptive FIR filters under computational constraint for adaptive tap-position control,'' \emph{IEEE Trans. Circuits Syst. II}, vol.~43, pp.~629-636, Sept. 1996.
\bibitem{Sugiyama3}
A.~Sugiyama, S.~Ikeda, and A.~Hirano, ``A fast convergence algorithm for sparse-tap adaptive FIR filters identifying an unknown number of dispersive regions,'' \emph{IEEE Trans. Signal Process.}, vol.~50, no.~12, pp.~3008-3017, December 2002.
\bibitem{Nosko}
O.~A.~Noskoski, J.~C.~M.~Bermudez, S.~J.~M.~Almeida, ``Region-based wavelet-packet adaptive algorithm for identification of sparse impulse responses,'' \emph{IEEE Trans. Signal Process}, vol.~61, no.~13, pp.~3321-3333, July, 2013.
\bibitem{Gu2}
Y.~Gu, J.~Jin, and S.~Mei, ``$l_0$ norm constraint LMS algorithm for sparse system identification,'' \emph{IEEE Signal Process. Lett.}, vol.~16, no.~9, pp.~774-777, Sep. 2009.
\bibitem{Chen}
Y.~Chen, Y.~Gu, and A.~O.~Hero, ``Sparse LMS for system identification,'' \emph{IEEE ICASSP}, pp.~3125-3128, Taiwan, Apr. 2009.
\bibitem{Taheri}
O.~Taheri and S.~A.~Vorobyov, ``Reweighted $l_1$-norm penalized LMS for sparse channel estimation and its analysis,'' Submitted to \emph{IEEE Trans. Signal Process.}.
\bibitem{Mohimani}
H.~Mohimani, M.~Babaie-Zadeh, and C.~Jutten, ``A fast approach for overcomplete sparse decomposition based on smoothed L0 norm,'' \emph{IEEE Trans. Signal Process.}, vol.~57, no.~1, pp.~289-301, Jan. 2009.
\bibitem{Mohimani2}
H.~Mohimani, M.~Babaie-Zadeh, I.~Gorodnitsky, and C.~Jutten, ``Sparse recovery using smoothed L0 (SL0): convergence analysis,'' \emph{ArXiv preprint arXiv:1001.5073}, 2010.
\bibitem{Wu}
F.~Wu and F.~Tong, ``Gradient optimization p-norm-like constraint LMS algorithm for sparse system estimation,'' \emph{Signal Process.}, 93, 967-971, 2013.
\bibitem{Wu2}
F.~Wu, Y.~Zhou, F.~Tong and R.~Kastner, ``Simplified p-norm-like constraint LMS algorithm for efficient estimation of underwater acoustic channels,'' \emph{Journal of Marine Science and Application}, Volume~12, Issue~2, pp.~228-234, June 2013.
\bibitem{Chen2}
Y.~Chen, Y.~Gu, and A.~O.~Hero, ``Regularized least-mean-square algorithms,'' \emph{ArXiv e-prints Dec. 2010} [Online]. Available: http://arxiv.org/abs/1012.5066v2.
\bibitem{Su2}
G.~Su, J.~Jin, Y.~Gu, and J.~Wang, ``Performance analysis of $l_0$ norm constraint least mean square algorithm,'' \emph{IEEE Trans. Signal Process.}, vol.~60, no.~5, pp.~2223-2235, May 2012.
\bibitem{Eldar}
Y.~C.~Eldar and M.~Mishali, ``Robust recovery of signals from a structured union of subspaces,'' \emph{IEEE Trans. Inf. Theory}, vol.~55, no.~11, pp.~5302-5316, Nov. 2009.
\bibitem{Stojnic}
M.~Stojnic, F.~Parvaresh, and B.~Hassibi, ``On the Reconstruction of Block-Sparse Signals With an Optimal Number of Measurements,'' \emph{IEEE Trans. Signal Process.}, vol.~57, no.~8, pp.~3075-3085, 2009.
\bibitem{Stojnic2}
M.~Stojnic, ``$l_2/l_1$-Optimization in Block-Sparse Compressed Sensing and Its Strong Thresholds,'' \emph{IEEE Journal of Selected Topics in Signal Processing}, vol.~4, no.~2, pp.~350-357, Apr. 2010.
\bibitem{Liu2}
J.~Liu, J.~Jin, and Y.~Gu, ``Efficient Recovery of Block Sparse Signals via Zero-point Attracting Projection,'' \emph{IEEE ICASSP}, pp.~3333-3336, Mar.~25-30, 2012, Kyoto, Japan.
\bibitem{Elhamifar}
E.~Elhamifar and R.~Vidal, ``Block-Sparse Recovery via Convex Optimization,'' \emph{IEEE Trans. Signal Process.}, vol.~60, no.~8, pp.~4094-4107, Aug. 2012.
\bibitem{Baraniuk}
R.~G.~Baraniuk, V.~Cevher, M.~F.~Duarte, and C.~Hegde, ``Model-based compressive sensing,'' \emph{IEEE Trans. Inf. Theory}, vol.~56, no.~4, pp.~1982-2001, 2010.
\bibitem{Eldar2}
Y.~C.~Eldar, P.~Kuppinger, and H.~Bolcskei, ``Block-sparse signals: uncertainty relations and efficient recovery,'' \emph{IEEE Trans. Signal Process.}, vol.~58, no.~6, pp.~3042-3054, 2010.
\bibitem{Cevher}
V.~Cevher, M.~F.~Duarte, C.~Hegde, and R.~G.~Baraniuk, ''Sparse signal recovery using Markov random fields,'' \emph{NIPS}, Vancouver, BC, Canada, Dec. 2008.
\bibitem{Ben-Haim}
Z.~Ben-Haim and Y.~C.~Eldar, ``Near-Oracle Performance of Greedy Block-Sparse Estimation Techniques From Noisy Measurements,'' \emph{IEEE Trans. Signal Process.}, vol.~5, no.~5, pp.~1032-1047, 2011.
\bibitem{Yu}
L.~Yu, H.~Sun, J.~P.~Barbot, and G.~Zheng, ``Bayesian compressive sensing for cluster structured sparse signals,'' \emph{Signal Process.}, vol.~92, no.~1, pp.~259-269, 2012.
\bibitem{Zhang}
Z.~Zhang and B.~D.~Rao, ``Extension of SBL Algorithms for the Recovery of Block Sparse Signals With Intra-Block Correlation,'' \emph{IEEE Trans. Signal Process.}, vol.~61, no.~8, pp.~2009-2015, 2013.
\bibitem{Cevher2}
V.~Cevher, P.~Indyk, C.~Hegde, and R.~G.~Baraniuk, ``Recovery of clustered sparse signals from compressive measurements,'' \emph{SAMPTA}, Marseille, France, May. 2009.
\bibitem{Parvaresh}
F.~Parvaresh and B.~Hassibi, ``Explicit measurements with almost optimal thresholds for compressed sensing,'' \emph{IEEE ICASSP}, Mar-Apr 2008.
\bibitem{Eksioglu}
E.~M.~Eksioglu, ``Group sparse RLS algorithms'', \emph{International Journal of Adaptive Control and Signal Processing}, Dec.~11, 2013.
\bibitem{Bradley}
P.~S.~Bradley and O.~L.~Mangasarian, ``Feature selection via concave minimization and support vector machines,'' \emph{ICML}, 1998, pp.~82-90.
\bibitem{McCoy}
B.~M.~McCoy and T.~T.~Wu. \emph{The two-dimensional Ising model.} Harvard Univ. Press, 1973.

\end{thebibliography}
\end{document}